\pdfoutput=1
\RequirePackage{ifpdf}
\ifpdf 
\documentclass[pdftex]{sigma}
\else
\documentclass{sigma}
\fi

\numberwithin{equation}{section}

\newtheorem{Theorem}{Theorem}[section]
\newtheorem*{Theorem*}{Theorem}
\newtheorem{Corollary}[Theorem]{Corollary}
\newtheorem{Lemma}[Theorem]{Lemma}

\theoremstyle{definition}

\newtheorem{Remark}[Theorem]{Remark}
\newtheorem{assumption}[Theorem]{Assumptions}

\usepackage[labelfont=bf,labelsep=period,font=small]{caption}
\usepackage{subcaption}

\newcommand{\R}{{\mathbb R}}

\newcommand{\C}{{\mathbb C}}

\newcommand{\re}{\operatorname{Re}}
\newcommand{\im}{\operatorname{Im}}

\newcommand{\RN}[1]{%
	\textup{\uppercase\expandafter{\romannumeral#1}}%
}

\begin{document}

\allowdisplaybreaks

\newcommand{\arXivNumber}{2510.00843}

\renewcommand{\PaperNumber}{106}

\FirstPageHeading

\ShortArticleName{Partition Functions of Two-Dimensional Coulomb Gases}

\ArticleName{Partition Functions of Two-Dimensional Coulomb\\ Gases with Circular Root- and Jump-Type\\ Singularities}

\Author{Kohei NODA}

\AuthorNameForHeading{K.~Noda}

\Address{Institut de Recherche en Math\'ematique et Physique, Universit\'e Catholique de Louvain,\\
Louvain-la-Neuve, 1348, Belgium}
\Email{\mail{kohei.noda@uclouvain.be}, \mail{kohei.noda@uclouvain.be}}
\URLaddress{\url{https://sites.google.com/view/koheinoda/home}}

\ArticleDates{Received January 15, 2026, in final form October 10, 2026; Published online October 30, 2026}

\Abstract{In this paper, we study the random polynomial $p_n(\rho):=\prod_{j=1}^n (|z_j|-\rho)$, where the points $\{z_j\}_{j=1}^n$ are the eigenvalues of random normal matrices with a radially symmetric potential. We establish precise large $n$ asymptotic expansions for the moment generating function \smash{$\mathbb{E}\bigl[{\rm e}^{(u \operatorname{Im}\log p_n(\rho))/\pi} {\rm e}^{a\operatorname{Re}\log p_n(\rho)}\bigr]$}, $u\in\mathbb{R}$, $a>-1$, where $\rho>0$ lies in the bulk of the spectral droplet. The asymptotic expansion is expressed in terms of parabolic cylinder functions, which confirms a conjecture of Byun and Charlier. This also provides the first free energy expansion of two-dimensional Coulomb gases with general circular root- and jump-type singularities. While the $a=0$ case has already been widely studied in the literature due to its relation to counting statistics, we also obtain new results for this special case.}


\Classification{41A60; 60B20; 60G55}

\section{Introduction and statement of results}\label{section: introduction}
For $\boldsymbol{z}_n=(z_1,\dots,z_n)\in\C^n$, consider the $n$-fold integral
\begin{align}
\label{def of det partition function}
Z_{n,u,a}[Q]&:=\int_{\C^n}\prod_{1\leq j<k\leq n}|z_{k} -z_{j}|^{2}\prod_{j=1}^n |z_{j}|^{2\alpha}{\rm e}^{-nQ(z_j)}\omega(z_j)\,\frac{\mathrm{d}^2z_j}{\pi}, \qquad\alpha>-1,
\end{align}
where $\mathrm{d}^2z$ is the Lebesgue measure on $\C$, and $Q\colon \C\to\R$ is called the external potential.
The factor $|z|^{2\alpha}={\rm e}^{\alpha\log|z|^2}$
is referred to as a \emph{Fisher--Hartwig singularity} at the origin, see e.g., \cite[Section~1.4]{ACC2023c}. The condition $\alpha>-1$ is necessary for the convergence of the integral in \eqref{def of det partition function}.
Our assumptions on $Q$ are stated in Assumptions~\ref{Assumption_Q} below.
Here, $\omega(z)\equiv\omega(z;u,a)$ possesses a~root singularity and a jump along the circle centered at 0 of radius $\rho>0$; more precisely, it is defined by
\begin{equation}
\label{def of omega}
\omega(z):=|x-\rho|^{a}\begin{cases}
 {\rm e}^{u}, & \mbox{if }x<\rho, \\
 1, & \mbox{if } x\geq \rho,
\end{cases}
\qquad x=|z|,\quad a>-1,\quad u\in\R.
\end{equation}
Integrals of the form \eqref{def of det partition function} are typically called partition functions in the literature. They find applications in random matrix theory and statistical physics, and have therefore been widely studied in \cite{AFLS25,ACC2023c,Byun25,BFL25,BKS2023,BKSY2025,BSY2025,LS2017,WebbWong}. $Z_{n,0,0}[Q]$ is also the normalization constant of the following probability measure:
\begin{align}\label{def of det point process}
d\mathbb{P}_{n}(\boldsymbol{z}_n):=\frac{1}{Z_{n,0,0}[Q]} \prod_{1 \leq j < k \leq n} |z_{k} -z_{j}|^{2} \prod_{j=1}^{n}|z_{j}|^{2\alpha}{\rm e}^{-n Q(z_{j})}\,\frac{{\rm d}^2z_j}{\pi},\qquad \alpha>-1,
\end{align}
which represents the joint probability distribution of the eigenvalues of a random normal matrix \cite{BFreview,Forrester,Mehta}.
In particular, the choice $Q(z)=|z|^2$ and $\alpha=0$ corresponds to the complex Ginibre ensemble \cite{Ginibre}.
The measure \eqref{def of det point process} is a {\it determinantal point process} \cite{Forrester}. For a review of recent developments on non-Hermitian random matrices, see \cite{BFreview}.

The ratio $\frac{Z_{n,u,a}[Q]}{Z_{n,0,0}[Q]}$ gives the joint moment generating function of $(\re\log p_n(\rho),\im\log p_{n}(\rho))$, where $p_n$ is the random polynomial given by
\begin{equation*}
p_n(x):=\prod_{j=1}^{n} \bigl(|z_j|-x\bigr).
\end{equation*}
Note that the roots of $p_n$ are the moduli $\{|z_j|\}_{j=1}^n$, where $\{z_j\}_{j=1}^{n}$ are distributed according to~\eqref{def of det point process}.
For $\rho>0$, we define
\[
\log p_n(\rho):=\log|p_n(\rho)|+\pi {\rm i}\,\mathrm{N}_\rho,\qquad
\mathrm{N}_\rho:=\#\{z_j\mid |z_j|<\rho\}.
\]
Here $\log$ denotes the principal branch of the logarithm.
If we denote by $\mathbb{E}$ the expectation with respect to~\eqref{def of det point process}, then
\begin{equation}
\label{def of mathcal En}
 \mathcal{E}_{n,u,a} := \mathbb{E}\bigl[{\rm e}^{\frac{u}{\pi}\im \log p_n(\rho)}{\rm e}^{a\re \log p_n(\rho)}\bigr] = \frac{Z_{n,u,a}[Q]}{Z_{n,0,0}[Q]},
\end{equation}
with $u \in \R$, $a > -1$.
The partition function $Z_{n,u,a}[Q]$ for $a\neq 0$ was first introduced in the work~\cite{BC2022} of Byun and Charlier, where the potential~$Q(z)$ in \eqref{def of det point process} was chosen to be the Mittag-Leffler potential $Q(z) = |z|^{2b}$ with $b > 0$.
In \cite[Theorem~1.1]{BC2022}, they established the large $n$ asymptotic expansion of $\mathcal{E}_{n,u,a}$ up to and including the term of order~1, with a precise estimate for the error term. However, only the case $a\in\mathbb{N}$ was considered.
An interesting phenomenon observed in their result is the appearance of associated Hermite polynomials in the coefficients of the terms of order~$\sqrt{n}$ and~$1$. The large~$n$ asymptotics of $\mathcal{E}_{n,u,a}$ in the general case $a>-1$ was left open in~\cite{BC2022}, and was conjectured to involve parabolic cylinder functions in place of the associated Hermite polynomials, see \cite[Remark~1.3]{BC2022}.
We also mention that the weight function~\eqref{def of omega} was studied in the study for the correlation kernel in \cite[equation~(2.4)]{Seo}.

\textit{In this paper, we obtain precise large $n$ asymptotics of $\mathcal{E}_{n,u,a}$ for general $u\in\R$, $a>-1$ and rotation-invariant $Q$, in the regime where $\rho$ lies in the bulk.} In particular, we confirm the conjecture from \cite{BC2022} that the asymptotics involve parabolic cylinder functions. This type of asymptotic behavior is completely new in random matrix theory, to the best of our knowledge. The case $a=0$ of our main results is already of interest, as it generalizes previous results on counting statistics of random normal matrix eigenvalues for general rotation-invariant potentials, as explained below.

For $a=0$, \eqref{def of mathcal En} reduces to the moment generating function of the disk counting statistics of random normal matrix eigenvalues, i.e.,
\begin{equation*}
 \mathcal{E}_{n,u,0}=\mathbb{E}\bigl[{\rm e}^{u\,{\rm N}_{\rho}}\bigr].
\end{equation*}
Counting statistics of random normal matrix eigenvalues have attracted considerable attention in recent years, see, e.g., \cite{ABE2023,ABES2023,ADM24,ACCL1,ACCL2,C2021-FH,CL2023,FenzlLambert,Lacroix2019,LMS2018,LeeRiser2016,SDMS2020,SDMS2021}.
In the work \cite{C2021-FH}, Charlier established the precise large $n$ asymptotic expansion of $\mathcal{E}_{n,u,0}$ for the Mittag-Leffler potential $Q(z)=|z|^{2b}$ with $b>0$. The asymptotics of the \textit{joint} moment generating functions, in the critical regime where all disk boundaries are merging at speed $n^{-1/2}$, were then obtained in the follow-up work \cite{CL2023}. In \cite{ACCL1,ACCL2}, Ameur, Charlier, Cronvall, and Lenells then treated the more difficult hard-edge regime where all disk boundaries are merging at speed $n^{-1}$ near a hard wall. In~\cite{ACCL2}, the asymptotics contain an oscillatory term due to the fact that the particles accumulate on several components. The work \cite{ACM2024} also treats a hard-edge case, but in a simpler situation where there is no bulk. Leading order asymptotics of $\mathcal{E}_{n,u,0}$ were then obtained in \cite{ABES2023} for general rotation invariant potentials, and leading order asymptotics of $\mathrm{Var}[{\rm N}_{\rho}]$ in \cite{MMO25} for general potentials and domains. Our main result, in the special case $a=0$, improves on \cite{ABES2023} by providing the next two terms in the large $n$ asymptotics of $\mathcal{E}_{n,u,0}$. Since counting statistics have attracted considerable attention in recent years, for the convenience of the reader this particular case is stated separately in Theorem~\ref{theorem:asymptotic expansion of counting statistics} below. In Corollary~\ref{corollary:cumulant of counting statistics}, we also provide precise large $n$ asymptotics of all cumulants of $\mathrm{N}_{\rho}$ (not just the variance).

We will make the following assumptions on $Q$:
\begin{assumption}\label{Assumption_Q}
We suppose that the potential is rotation invariant, i.e., $Q(z)=q(|z|)$, and satisfies the following conditions:
\begin{enumerate}\itemsep=0pt
\item[\textup{(1)}] 
\smash{$\liminf_{|z|\to\infty}\frac{Q(z)}{2\log|z|}>1$}, which guarantees that $Z_{n,u,a}[Q]<+\infty$.
\item[\textup{(2)}] 
$Q$ is $C^{6}$-smooth in a neighborhood of the droplet, globally subharmonic in $\C$, and strictly subharmonic (i.e., $\Delta Q(z)>0$) in a neighborhood of the droplet (the droplet is defined in~\eqref{def of droplet} below).
\end{enumerate}
\end{assumption}
\begin{Remark}\label{remark:MF Potential}
Assumptions~\ref{Assumption_Q} cover a broad class of rotation-invariant potentials. However, they do not cover the power-law potentials \(Q(z)=|z|^{2b}\) with \(b\neq1\), which form a distinguished subclass of the Mittag-Leffler potentials. Indeed,
\[
\Delta Q(z)\asymp |z|^{2b-2}, \qquad z\to0,
\]
and hence \(\Delta Q(z)\to0\) for \(b>1\), whereas \(\Delta Q(z)\to\infty\) for \(b\in(0,1)\). This power-law behavior of the Laplacian is associated with a conical singularity at the origin; see, e.g., \cite{CCLW16, Kalvin2021}.

Such singularities occur naturally in the study of quantum Hall states on surfaces with conical points. In the presence of a constant magnetic field, the corresponding Coulomb-gas partition function naturally involves both conical and Fisher--Hartwig singularities, as in \eqref{def of det partition function}. Partition functions of this type admit descriptions in terms of Barnes double zeta functions~\cite{Kl16} and are closely related to \(\zeta\)-regularized determinants of Laplace operators on surfaces with conical singularities~\cite{Kalvin2021}. It would be of interest to extend the results of the present paper to classes of potentials allowing for such singular behavior at the origin; we do not pursue this direction here.

Moreover, Assumption~\ref{Assumption_Q}\,(2) excludes the disconnected-droplet considered in \cite{ACC2023a, ACC2023c,AC2026}.
\end{Remark}
Under Assumptions~\ref{Assumption_Q}, the empirical measure $\frac1n\sum_{j=1}^n\delta_{z_j}$ of \eqref{def of det point process} converges weakly to the measure $\sigma_Q$ given by
\begin{equation*}
{\rm d}\sigma_Q := \Delta Q \cdot \mathbf{1}_{S}\,\frac{{\rm d}^2z}{\pi},
\end{equation*}
where $S \equiv S_{Q}$ is a compact subset of $\C$ called the {\it droplet}.
Here we use $\Delta=\frac{1}{4}(\partial_x^2+\partial_y^2)$ for $z=x+{\rm i}y$.
Under parts (1) and (2) of Assumptions~\ref{Assumption_Q}, $S$ is of the form
\begin{equation}\label{def of droplet}
S = \mathbb{A}_{r_0,r_1}:=\{z \in \C \mid r_0 \leq |z| \leq r_1\},
\end{equation}
where $r_0$ is the largest solution to $rq'(r)=0$, and $r_1$ is the smallest solution of $r q'(r)=2$; see, e.g., \cite{BKS2023} and \cite[Section~IV.6]{SaTo}.
Throughout this paper, we assume that $r_0=0$, i.e., $S=\mathbb{D}_{r_1}:=\{z\in\C\mid |z|\leq r_1\}$. (The analysis done in this paper can be adapted to the case $r_0>0$, which is in fact simpler.)

We recall that the complementary error function is defined by
\begin{equation*}
\mathrm{erfc}(t):=\frac{2}{\sqrt{\pi}}\int_t^{+\infty}{\rm e}^{-x^2}\,{\rm d}x, \qquad t\in\R,
\end{equation*}
see, e.g., \cite[equation~(7.2.1)]{NIST}.
Following \cite[equation~(1.6)]{C2021-FH}, we introduce
\begin{equation}
\label{def of mathcalFts}
\mathcal{F}(t,s):=\log\Bigl(1+\frac{s-1}{2}\mathrm{erfc}(t)\Bigr), \qquad
t\in \R, \quad s \in \C\backslash (-\infty, 0],
\end{equation}
where the principal branch is chosen for the $\log$.

Our first main result corresponds to the case $a=0$.
\begin{Theorem}[counting statistics]\label{theorem:asymptotic expansion of counting statistics}
Let $\rho\in(0,r_1)$, $\alpha>-1$, $u\in\R$, and $a=0$.
Under Assumptions~{\rm \ref{Assumption_Q}}, for any fixed $x\in\R$, there exists $\delta>0$ such that, as $n\to+\infty$, we have
\begin{equation}\label{def of calEnu0 asymptotics}
\log\mathcal{E}_{n,u,0}
=C_1(u) n+C_2(u) \sqrt{n}+C_3(u)+\mathcal{O}\left(\frac{(\log n)^3}{n^{1/12}}\right),
\end{equation}
uniformly for $u\in\{z\in\C\mid |z-x|\leq\delta\}$, where
\begin{align*}
&C_1(u):=
u\int_{\mathbb{D}_{\rho}}{\rm d}\sigma_Q(z),
\\
&C_2(u):=
\rho\sqrt{2\Delta Q(\rho)}\int_{0}^{+\infty}
\bigl(
\mathcal{F}(x,{\rm e}^{u})
+
\mathcal{F}(x,{\rm e}^{-u})
\bigr)\,{\rm d}x,
\\
&C_3(u):=
-
\left(\alpha+\frac{1}{2}\right)u
+\frac{1}{6}\left(2+\frac{\rho \partial_r\Delta Q(\rho)}{\Delta Q(\rho)}\right)u
\nonumber\\
&\hphantom{C_3(u):=}{}
+\frac{1}{3}\left(2+\frac{\rho\partial_r\Delta Q(\rho)}{\Delta Q(\rho)}\right)
\int_{0}^{+\infty}
x\bigl(
\mathcal{F}(x,{\rm e}^{u})
-
\mathcal{F}(x,{\rm e}^{-u})
\bigr)\,{\rm d}x.
\end{align*}
\end{Theorem}

\begin{Remark}[{consistency with \cite[Theorem 1.1]{C2021-FH}}]\label{remark:counting statistics}
First, note that $C_3(u)$ is given in a simpler form than $C_3$ in \cite[Theorem~1.1]{C2021-FH}.

Second, as mentioned in Remark~\ref{remark:MF Potential}, Assumptions~\ref{Assumption_Q} do not cover the case $Q(z)=|z|^{2b}$ for $b\neq1$. However, surprisingly,
substituting $Q(z)=|z|^{2b}$ into Theorem~\ref{theorem:asymptotic expansion of counting statistics} recovers \cite[Theorem~1.1]{BC2022} for any $b>0$, see Appendix~\ref{section: appendix counting}.
It would be interesting to extend the results of the present paper to a broader class of potentials, including $Q(z)=|z|^{2b}$ for $b>0$.
\end{Remark}

Recall that the cumulants $\{\kappa_j\}_{j\in\mathbb{N}_{>0}}$ of the random variable ${\rm N}_{\rho}$ (see \eqref{def of mathcal En} above) are defined through the expansion
\[
\log \mathbb{E}\bigl[{\rm e}^{u\,{\rm N}_{\rho}}\bigr]
=
\kappa_1u+\frac{\kappa_2u^2}{2!}+\frac{\kappa_3u^3}{3!}+\cdots,\qquad u\to 0,
\]
or equivalently by
\begin{equation*}
\kappa_j=\partial_u^j\log \mathbb{E}\bigl[{\rm e}^{u\,{\rm N}_{\rho}}\bigr]\Bigr|_{u=0},\qquad j\in\mathbb{N}.
\end{equation*}
In particular, since $\mathbb{E}\bigl[{\rm e}^{u\,{\rm N}_{\rho}}\bigr]$ is analytic for $u\in\C$ and positive for any $u\in\R$, and Theorem~\ref{theorem:asymptotic expansion of counting statistics} is valid uniformly for $u\in\{z\in\C\mid |z-x|\leq \delta\}$ for some $\delta>0$, Cauchy's formula implies that for any $u\in\R$, $j\in\mathbb{N}$, we have
\begin{align*}
&
\partial_u^{j}\bigl\{\log \mathbb{E}\bigl[{\rm e}^{u {\rm N}_{\rho}}\bigr]-\bigl(C_1(u) n+C_2(u) \sqrt{n}+C_3(u)\bigr)\bigr\}
\\
&\qquad =
\frac{j!}{2\pi {\rm i}}\oint_{|\zeta-u|=\frac{\delta}{2}}
\frac{\log \mathbb{E}\bigl[{\rm e}^{\zeta {\rm N}_{\rho}}\bigr]-\bigl(C_1(\zeta) n+C_2(\zeta) \sqrt{n}+C_3(\zeta)\bigr)}{(\zeta-u)^{j+1}}\,{\rm d}\zeta\\
& \qquad =\mathcal{O}\left(\frac{(\log n)^3}{n^{1/12}}\right),\qquad n\to+\infty.
\end{align*}
The above equation shows that \eqref{def of calEnu0 asymptotics} can be differentiated with respect to $u$ any fixed number of times without increasing the error term.
Therefore, Theorem~\ref{theorem:asymptotic expansion of counting statistics} combined with \cite[Corollary~1.6]{C2021-FH} provides the following corollary, generalizing the result in \cite[Proposition~2.14]{ABES2023}.

\begin{Corollary}
\label{corollary:cumulant of counting statistics}
Let $\rho\in(0,r_1)$, $\alpha>-1$, $u\in\R$, and $a=0$.
Under Assumptions~{\rm \ref{Assumption_Q}}, as $n\to+\infty$, we have
\begin{equation*}
\kappa_j=\begin{cases}
\displaystyle C_1'(0) n+C_3'(0)+\mathcal{O}\left(\frac{(\log n)^3}{n^{1/12}}\right) & \mbox{if } j=1, \vspace{1mm}\\
\displaystyle \partial_u^jC_3(0)+\mathcal{O}\left(\frac{(\log n)^3}{n^{1/12}}\right) & \mbox{if } \text{$j\neq 1$ is odd}, \vspace{1mm}\\
\displaystyle \partial_u^jC_2(0) \sqrt{n}+\mathcal{O}\left(\frac{(\log n)^3}{n^{1/12}}\right) & \mbox{if } \text{$j$ is even}.
\end{cases}
\end{equation*}
\end{Corollary}

The central limit theorem for the disk counting statistic $\mathrm{N}_{\rho}$ follows from Corollary~\ref{corollary:cumulant of counting statistics} by the same argument as in \cite[proof of Corollary~1.6\,(e)]{C2021-FH}.

We now turn our attention to \eqref{def of mathcal En} in the general case $a>-1$.
For any $x\in\R$, $u\in\C$, and $a>-1$, define
\begin{equation}
\label{def of cal H au}
\mathcal{H}_{a,u}(x):=\frac{\Gamma(a+1)}{\sqrt{2\pi}}{\rm e}^{-\frac{1}{4}x^2}g_{a,u}(x),
 \qquad g_{a,u}(x):=
{\rm e}^{u}D_{-a-1}(x)+D_{-a-1}(-x),
\end{equation}
where $D_{-\nu}(x)$ is the parabolic cylinder function \cite[equation~(12.5.1)]{NIST}, which is defined by
\begin{equation}
\label{def of parabolic cylinder function}
 D_{-\nu}(z)=U\bigl(\nu-\tfrac{1}{2},z\bigr),\qquad \re\nu>0.
\end{equation}
Here,
\begin{equation*}
 U(\nu,z):=\frac{{\rm e}^{-\frac{1}{4}z^2}}{\Gamma\bigl(\nu+\frac{1}{2}\bigr)}
 \int_{0}^{+\infty}{\rm e}^{-zt}t^{\nu-\frac{1}{2}}{\rm e}^{-\frac{1}{2}t^2}\,{\rm d}t,\qquad \re \nu>-\frac{1}{2}.
\end{equation*}
For the relationship between \eqref{def of parabolic cylinder function} and the associated Hermite polynomials, see Appendix~\ref{section:appendix parabolic cylinder function}.
Note that for any $a>-1$, the integrand of $D_{-a-1}(x)$ is positive for all $x\in\R$.
Consequently, $\mathcal{H}_{a,u}(x)$ is positive for all $x\in\R$ whenever $a>-1$ and $u\in\R$.
Therefore, it follows that the logarithm $\log \mathcal{H}_{a,u}(x)$ is well defined for all $x\in\R$ and $u\in\R$.

The following is the main result of this paper.
\begin{Theorem}[counting statistics and root-type statistics]\label{theorem:calEn expansion}
Let $\rho\in(0,r_1)$, $\alpha>-1$, $u\in\R$, and $a>-1$.
Under Assumptions~{\rm \ref{Assumption_Q}}, for any fixed $x\in\R$, there exists $\delta>0$ such that as $n\to+\infty$, we have
\begin{equation}\label{def of expansion of calEnua}
\log\mathcal{E}_{n,u,a}=C_1(u,a) n+C_2(u,a) \sqrt{n}+C_3(u,a)+\mathcal{O}\left(\frac{(\log n)^3}{n^{1/12}}\right),
\end{equation}
uniformly for $u\in\{z\in\C\mid |z-x|\leq \delta\}$ and $a$ in compact subsets of $(-1,+\infty)$, where
\begin{gather}
C_1(u,a):=
\int_{\C}\log\omega(z)\,{\rm d}\sigma_{Q}(z)
\nonumber\\
\hphantom{C_1(u,a)} \,=
\int_{\mathbb{D}_{\rho}} (u+a\log(\rho-|z|) )\, {\rm d}\sigma_Q(z)
+\int_{\mathbb{A}_{\rho,r_1}}a\log (|z|-\rho)\,{\rm d}\sigma_Q(z),\label{def of mathcal En constant C1}
\\
\label{def of mathcal En constant C2}
C_2(u,a) :=
\rho\sqrt{\Delta Q(\rho)}
\int_{-\infty}^{+\infty}
\bigl(\log[\mathcal{H}_{a,u}(x)]
-a\log|x|-u\mathbf{1}_{(-\infty,0)}(x)
\bigr)
\,{\rm d}x,
\\
C_3(u,a) :=
-\frac{a}{2}\log\left(\frac{r_1}{\rho}-1\right)
-\frac{a(a-1)}{4}
\frac{r_1}{r_1-\rho}
\nonumber
\\
\hphantom{C_3(u,a) :=}{}
-
\frac{a}{4}
\int_{0}^{r_1}
\frac{1}{x-\rho}
\left(
\frac{x\partial_x\Delta Q(x)}{\Delta Q(x)}
-
\frac{\rho\partial_r\Delta Q(\rho)}{\Delta Q(\rho)}
\right)
{\rm d}x
\nonumber
\\
\hphantom{C_3(u,a) :=}{}
+
\frac{a}{4}\left(
4\alpha+a+2-\frac{\rho\partial_r\Delta Q(\rho)}{\Delta Q(\rho)}
\right)\log\left(\frac{r_1}{\rho}-1\right)-
\left(\alpha+\frac{1}{2}\right)u\nonumber
\\
\hphantom{C_3(u,a) :=}{}-
\frac{a}{12}\left(1-\frac{\rho\partial_r\Delta Q(\rho)}{\Delta Q(\rho)}\right)u
+\frac{1}{6}\left(2+\frac{\rho \partial_r\Delta Q(\rho)}{\Delta Q(\rho)}\right)u
+\frac{1}{6}\left(2+\frac{\rho\partial_r\Delta Q(\rho)}{\Delta Q(\rho)}\right)
\nonumber
\\
\hphantom{C_3(u,a) :=}{}
\times
\int_{-\infty}^{+\infty}
\left[
x\bigl(\log [\mathcal{H}_{a,u}(x) ]
- u\mathbf{1}_{(-\infty,0)}(x)
\bigr)
-ax\log|x|
-\frac{a(a-1)x}{2(x^2+1)}
\right] {\rm d}x.\label{def of mathcal En constant C3}
\end{gather}
\end{Theorem}

\begin{Remark}\label{remark for theorem cal EN 1}
Again, as explained in Remark~\ref{remark:MF Potential}, Assumptions~\ref{Assumption_Q} do not cover the case $Q(z)=|z|^{2b}$ if $b\neq1$.
However, in a similar way as in Remark~\ref{remark:counting statistics}, substituting $Q(z)=|z|^{2b}$ into Theorem~\ref{theorem:calEn expansion} recovers \cite[Theorem~1.1]{BC2022} for any $b>0$, see Appendix~\ref{section:appendix parabolic cylinder function}.
\end{Remark}

\begin{figure}[t]\centering
	\begin{subfigure}{0.44\textwidth}
\centering
 		\includegraphics[width=\textwidth]{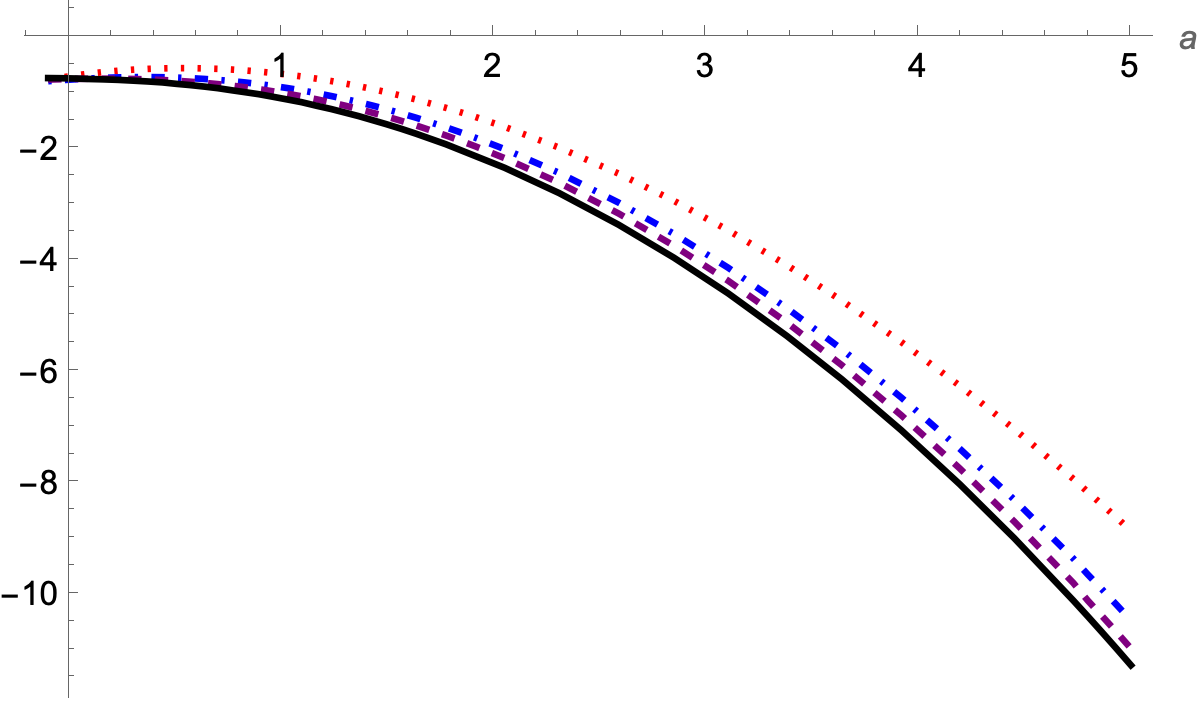}
 \subcaption{$a\mapsto C_3$}
 \end{subfigure}	 \qquad
	\begin{subfigure}{0.44\textwidth}
\centering
		\includegraphics[width=\textwidth]{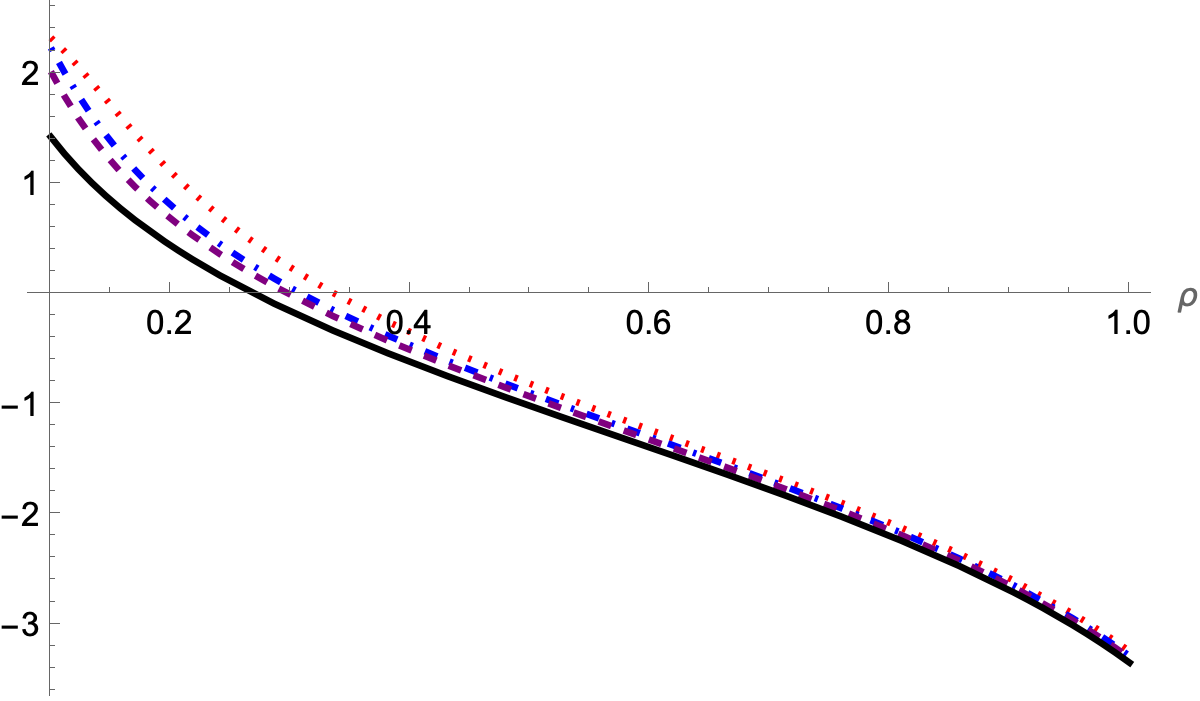}
 \subcaption{$\rho\mapsto C_3$}
\end{subfigure}	
	\caption{Plot (a) shows $a\mapsto C_3$ (black line) and its comparison $a\mapsto \log\mathcal{E}_{n,u,a}-\bigl(C_1 n+C_2 \sqrt{n}\bigr)$, where $Q(z)=0.2|z|^2+0.2345|z|^3$, $\alpha=0.667$, $u = 1.56$, $\rho=0.71r_1$, $n=10$ (red, dotted line), $n = 40$ (blue, dot-dashed line), and $n = 160$ (purple, dashed line).
Plot~(b) shows $\rho\mapsto C_3$ (black line) and its comparison $\rho\mapsto \log\mathcal{E}_{n,u,a}-(C_1 n+C_2 \sqrt{n})$, where $Q$, $\alpha$, $u$ are same as before, $a=1.25$, and $n=100$ (red, dotted line), $n = 300$ (blue, dot-dashed line), and $n = 600$ (purple, dashed line). This numerical computation is based on the fact that, as shown in Section~\ref{section:global analysis}, \eqref{def of expansion of calEnua} can be rewritten in terms of \eqref{def of mathcal En}.} \label{Fig_Comparison}
\end{figure}
As mentioned earlier, Theorem~\ref{theorem:calEn expansion} confirms the conjecture stated in \cite[Remark~1.3]{BC2022}.
Moreover, we obtained $C_3(u,a)$ in a simpler form than $C_3$ in \cite[Theorem~1.1]{BC2022}.
Figure~\ref{Fig_Comparison} provides a~numerical comparison to verify the asymptotic formula \eqref{def of mathcal En constant C3}.
We currently do not have a~conformal theoretic or geometric interpretation for \eqref{def of mathcal En constant C1}, \eqref{def of mathcal En constant C2}, and~\eqref{def of mathcal En constant C3} as in \cite{BKS2023, ZW2006}.
We also believe that the error estimate in~\eqref{def of expansion of calEnua} is not optimal, but we do not pursue this issue here. It is expected that the error can be improved to $n^{-1/2}$, see \cite[Theorem~1.1 and Remark~1.2]{BC2022}.

We finally deduce the large $n$ asymptotics of \eqref{def of det partition function} by combining Theorem~\ref{theorem:calEn expansion} with \cite[Theorem~1.4]{ACC2023c}.

We define
\begin{equation*}
I_{Q}[\mu]:=\int_{\C^2}\log \frac{1}{|z-w|}\,{\rm d}\mu(z){\rm d}\mu(w)+\int_{\C}Q(z)\,{\rm d}\mu(z),
\end{equation*}
which is referred to as the weighted logarithmic energy, defined over all compactly supported Borel probability measures~$\mu$.
In particular, if~$Q$ is lower semi-continuous and finite on a set of positive capacity, Frostman's theorem guarantees the existence of a unique equilibrium measure~$\sigma_Q$ minimizing the weighted logarithmic energy, see~\cite{SaTo}.
In particular, for the potential~$Q$ satisfying Assumptions~\ref{Assumption_Q}, we have
\[
I_{Q}[\sigma_Q] =q(r_1)-\log r_1-\frac{1}{4}\int_{0}^{r_1}rq'(r)^2\,{\rm d}r.
\]
We define
\begin{align*}
&E_{Q}[\sigma_{Q}]:=\int_{\C}\log \Delta Q\,{\rm d}\sigma_{Q},
\\
&F_{Q}[\sigma_{Q}]:=\frac{1}{12}\log\frac{1}{r_1^2\Delta Q(r_1)}-\frac{1}{16}\frac{r_1\partial_r\Delta Q(r_1)}{\Delta Q(r_1)}+\frac{1}{24}\int_0^{r_1}\left(\frac{\partial_r\Delta Q(r)}{\Delta Q(r)}\right)^2r\,{\rm d}r,
\end{align*}
and following \cite[equation~(1.14)]{ACC2023c}, we also define, for \(\ell_\alpha(z):=2\alpha\log|z|\),
\begin{align*}
\mathsf{e}_{\ell_{\alpha}}
&:=
\frac{1}{2}\int_{S}\ell_{\alpha}(z)\Delta \log \Delta Q(z)\,{\rm d}A(z)+\frac{1}{8\pi}\int_{\partial S}\partial_{\mathrm{n}}\ell_{\alpha}(z)\,|{\rm d}z|-\frac{1}{8\pi}\int_{\partial S}\ell_{\alpha}(z)\frac{\partial_{\mathrm{n}}\Delta Q(z)}{\Delta Q(z)}\,|{\rm d}z|,
\end{align*}
where ``$\partial_{\mathrm{n}}$'' designates differentiation in the normal direction to $\partial S$ pointing out from the droplet~$S$.
Here, $E_{Q}[\sigma_Q]$ represents the negative entropy of the equilibrium measure, while $F_{Q}[\sigma_Q]$ can be interpreted in terms of $\zeta$-regularized determinants associated with certain pseudo-differential operators, see \cite{BKS2023}, \cite[p.~284, Remark]{ACC2023c}, and the references therein.

For smooth linear statistics, the asymptotic cumulant generating function, including its mean correction, is implicit in the physical predictions of Zabrodin and Wiegmann \cite{ZW2006} and was rigorously established by Ameur, Hedenmalm, and Makarov \cite{AHM2015} in a general connected-droplet setting. Ameur and Cronvall \cite{AC2026} extended this fluctuation theory to a class of disconnected droplets separated by ring-shaped spectral gaps, see also \cite{ACC2023c}.
The free energy expansion with a~Fisher–Hartwig insertion at the origin is provided by \cite[Theorem 1.4]{ACC2023c}, where the quantity \(\mathsf e_{\ell_\alpha}\) above is obtained.

We conclude with the following result.
\begin{Theorem}[partition function with circular- and root-type singularities]
\label{corollary:partition function}
Let $\rho\in(0,r_1)$, $\alpha>-1$, $u\in\R$, and $a>-1$.
Under Assumptions~{\rm \ref{Assumption_Q}}, for any fixed $x\in\R$, there exists $\delta>0$ such that as $n\to+\infty$, we have
\[
\log Z_{n,u,a}[Q]
=
\widetilde{C}_1n^2+\widetilde{C}_2n\log n+\widetilde{C}_3n+\widetilde{C}_4\sqrt{n}+\widetilde{C}_5\log n+\widetilde{C}_6+\mathcal{O}\left(\frac{(\log n)^3}{n^{1/12}}\right),
\]
uniformly for $u\in\{z\in\C\mid |z-x|\leq\delta\}$ and $a$ in compact subsets of $(-1,+\infty)$, where
\begin{align*}
&\widetilde{C}_1:=-I_Q[\sigma_Q],
\qquad
\widetilde{C}_2:=\frac{1}{2},
\\
&\widetilde{C}_3:=\frac{\log 2\pi}{2}-1-\frac{E_Q[\sigma_Q]}{2}+\int_{\C}\ell_{\alpha}(z)\,{\rm d}\sigma_Q(z)+C_1(u,a),
\\
&\widetilde{C}_4:=C_2(u,a),
\qquad
\widetilde{C}_5:=\frac{5}{12}+\frac{\alpha^2}{2},
\\
&\widetilde{C}_6:=\zeta'(-1)-\log G(1+\alpha)+F_Q[\sigma_Q]+\frac{1+\alpha}{2}\log(2\pi)+\mathsf{e}_{\ell_{\alpha}}
 -\frac{\alpha}{2}\\
 &\hphantom{\widetilde{C}_6:=}{}
 +\frac{\alpha^2}{2}\log\bigl(r_1^2\Delta Q(0)\bigr)+C_3(u,a),
\end{align*}
where $C_1(u,a)$, $C_2(u,a)$, $C_3(u,a)$ are given by~\eqref{def of mathcal En constant C1}, \eqref{def of mathcal En constant C2}, \eqref{def of mathcal En constant C3}, respectively, and $\zeta(z)$ is the Riemann zeta function.
\end{Theorem}

\begin{Remark}
The coefficient $\widetilde{C}_5$ can be rewritten as $\frac{6-\chi}{12}+\frac{\alpha^2}{2}$, where $\chi$ is the Euler characteristic of the droplet. In our case, since the droplet is a disk, $\chi=1$. The fact that $\widetilde{C}_5$ can be rewritten in this way is consistent with the general conjecture~\cite{ZW2006}.
Similar results have been established in~\cite{BKS2023} for centered disks and annuli, in~\cite{ACC2023c} for multiple disjoint connected annuli possibly including a central disk, and in~\cite{BKSY2025} for the (generalized) spherical ensemble.
\end{Remark}

\begin{Remark}
Throughout this paper, we do not address the case where $r_0>0$ (see \eqref{def of droplet}). Nevertheless, the asymptotic expansion of $Z_{n,u,a}[Q]$ in this case can be established by arguments similar to those in \cite[Section~5]{AFLS25} and \cite[Section~5.2]{AL25}. To avoid introducing additional notation and technicalities, we do not pursue this direction here.
\end{Remark}

\subsubsection*{Related works}
We conclude by briefly mentioning recent developments in the literature related to Theorems~\ref{theorem:asymptotic expansion of counting statistics},~\ref{theorem:calEn expansion}, and~\ref{corollary:partition function}.

The study of structured determinants with singularities in one-dimensional point processes has a long history. In the seminal work~\cite{FisherHartwig}, Fisher and Hartwig conjectured the asymptotic behavior of Toeplitz determinants in the large-size limit when the weight function is supported on the unit circle and possesses root- and jump-type singularities, which is now commonly referred to as Fisher--Hartwig singularities.
For early works and more historical background, we refer to~\cite{DIKreview} and references therein.

Recent results on structured determinants with singularities in connection with random matrix theory include \cite{BCL2024,CGMY2020,CK2015, DIK, DIK14, Fahs} for Toeplitz determinants, \cite{CM2023, DXZ2022,DXZ2023,LYZ24,Xetal,XZZ2024,YZ24} for Fredholm determinants, \cite{BerWebbWong, Charlier, CFWW2021,ItsKrasovsky, Krasovsky} for Hankel determinants, \cite{BasorEhrhardt1, DIK} for Toeplitz+Hankel determinants.

There has been significant progress on the precise large-$n$ asymptotics of partition functions, moment generating functions, and hole probabilities for the random normal matrices mentioned earlier.
A seminal development is the works of Charlier \cite{C2021-FH,C2021} on the moment generating function and hole probabilities for random normal matrices with Mittag-Leffler potentials ${Q(z)=|z|^{2b}}$ for $b>0$.

Further progress was then made by Byun, Kang, and Seo~\cite{BKS2023}, who obtained precise large-$n$ asymptotics for the partition functions of two-dimensional Coulomb gases with rotation-invariant potentials.
Subsequent developments include the study of multi-component droplets~\cite{ACC2023c} by Ameur, Charlier, and Cronvall, who removed the global subharmonicity assumption on the potential and established new asymptotic results for droplets with multiple connected components, thereby extending the results of \cite{ZW2006}.

Other developments concern hole probabilities \cite{BP24, C2023}, partition functions with a fixed hard wall \cite{AFLS25}, and potentials for which~$\Delta Q$ vanishes along a circle inside the droplet~\cite{AL25}. Related free energy expansions have been obtained for spherical Coulomb gases \cite{BKSY2025,ZW2006} and for lemniscate-type potentials with a conical singularity at the origin~\cite{Byun25}.

Precise large-$n$ asymptotics for the partition function of the Ginibre ensemble with a big point-charge insertion that breaks rotational symmetry were obtained in~\cite{BSY2025}. This work is motivated by strong asymptotics of planar orthogonal polynomials for the Ginibre ensemble with a point insertion~\cite{BBLM2015}. For related developments, see \cite{BKP2023,BEG18,Byun,BLY2022,BY2023,LeeYang,LeeYang2,LY2023}, and for connections with Gaussian multiplicative chaos, see \cite{WebbWong}.

Functional central limit theorems for hard-edge disk counting statistics in Mittag-Leffler ensembles ($Q(z)=|z|^{2b}$) with a hard wall were obtained in \cite{Berezin2023}. For point insertions in other ensembles, including the truncated unitary and spherical ensembles, see \cite{BFKL25,BFL25, DKMS25}. In contrast to point insertions, the factor \eqref{def of omega} introduces circular root- and jump-type singularities. As in~\cite{BC2022}, our analysis does not require Riemann--Hilbert methods.
Free energy asymptotics for determinantal Coulomb gases on compact Riemann surfaces were studied in~\cite{Burgoin2025}.

\subsubsection*{Outline}
It suffices to prove Theorem~\ref{theorem:calEn expansion}, as Theorems~\ref{theorem:asymptotic expansion of counting statistics} and~\ref{corollary:partition function} follow as special cases.
The strategy for proving Theorem~\ref{theorem:calEn expansion} combines the Laplace method, the precise Riemann sum approximation of~\cite{C2021}, and the decomposition of a large sum into global and local analysis parts as in~\cite{BC2022}.

We first split the logarithmic sum in~\eqref{def of mathcal En} into a global part and a local part.
In Section~\ref{section:global analysis}, we establish the large $n$ asymptotics of the global part (Lemmas~\ref{lemma:S0}, \ref{lemma:S1}, and \ref{lemma:S3}),
where the proofs of Lemmas~\ref{lemma:S1} and~\ref{lemma:S3} require careful analysis of the error.
In Section~\ref{section:local analysis part}, we establish the large $n$ asymptotics of the summand corresponding to the local part.
This is first stated in Lemma~\ref{lemma:S2in part1}, and subsequently, each sum is approximated using Lemma~\ref{lemma:Euler--Maclaurin} with more refined error estimates than in the global part.

\section{Global analysis part for the proof of Theorem~\ref{theorem:calEn expansion}}
\label{section:global analysis}
Since $|z|^{2\alpha}{\rm e}^{-nq(|z|)}\omega(z)$ is rotation invariant, $\mathcal{E}_{n}\equiv\mathcal{E}_{n,u,a}$ can be identically expressed in terms of one-fold integrals.
This fact is well known as {\it Andr\'{e}ief identity} and has already been used in different contexts, see, e.g., \cite{ACC2023c,ACCL1,ACCL2,ACM2024,BC2022,BKS2023,BP24,C2021-FH}.
For fixed $u\in\R,a\in(-1,+\infty)$, and $\rho\in(0,r_1)$, we have
\begin{equation}
\label{def of calEn}
 \mathcal{E}_{n}\equiv\mathbb{E}\bigl[{\rm e}^{\frac{u}{\pi}\im\log p_n(\rho)}{\rm e}^{a\re \log p_n(\rho)}\bigr]
 =\frac{\mathcal{D}_n}{\mathcal{Z}_{n}},
\end{equation}
 where
\begin{align}
\label{def of Dn}
 \mathcal{D}_n&:=\prod_{j=0}^{n-1} \int_0^{+\infty} 2v^{2j+2\alpha+1}{\rm e}^{-nq(v)}\omega(v)\,{\rm d}v,
 \\
 \nonumber
 \mathcal{Z}_{n}&:=\prod_{j=0}^{n-1} \int_0^{+\infty} 2v^{2j+2\alpha+1}{\rm e}^{-nq(v)}\,{\rm d}v.
\end{align}
Here, $\omega(|z|)$ is defined by \eqref{def of omega}.
Note that \eqref{def of Dn} is simply written as
\begin{align*}
 \mathcal{D}_n&=
 \prod_{j=0}^{n-1}
 \biggl(
{\rm e}^{u}\int_0^{\rho} 2v^{2j+2\alpha+1}{\rm e}^{-nq(v)}(\rho-v)^{a}\,{\rm d}v
+
\int_{\rho}^{+\infty} 2v^{2j+2\alpha+1}{\rm e}^{-nq(v)}(v-\rho)^{a}\,{\rm d}v
 \biggr).
\end{align*}
In order to analyze the precise large $n$-asymptotics of \eqref{def of calEn}, we follow the robust strategy done in \cite{ACC2023c,BC2022,BKS2023}.
We define
\begin{equation}
\label{def of V tau}
 V_{\tau}(r):=q(r)-2\tau\log r, \qquad \tau\equiv \tau(j):=\frac{j}{n}.
\end{equation}
By differentiating \eqref{def of V tau} with respect to $r$, from \cite[equation~(2.4)]{BKS2023}, we have
\begin{align}
&V_{\tau}'(r)=q'(r)-\frac{2\tau}{r},\qquad V_{\tau}''(r)=4\Delta Q(r)-\frac{1}{r}V_{\tau}'(r), \nonumber\\
&V_{\tau}^{(3)}(r)=4\partial_{r}\Delta Q(r)-\frac{4}{r}\Delta Q(r)+\frac{2}{r^2}V_{\tau}'(r), \nonumber\\
&V_{\tau}^{(4)}(r)=4\partial_r^2\Delta Q(r)+\frac{12}{r^2}\Delta Q(r)-\frac{4}{r}\partial_{r}\Delta Q(r)-\frac{6}{r^3}V_{\tau}'(r).\label{def of V tau relationship 1}
\end{align}
For $0\leq \tau\leq 1$ we denote by $r_{\tau}$ the solution to
\begin{equation*}
 r_{\tau} q'(r_{\tau})=2\tau,
\end{equation*}
which gives rise to $V_{\tau}'(r_{\tau})=0$.
Note that $r_{\tau}$ satisfies the following differential equation \cite[equation~(2.6)]{BKS2023}:
\begin{equation}\label{def of rtau ODE}
\frac{{\rm d}r_{\tau}}{{\rm d}\tau}=\frac{1}{2r_{\tau}\Delta Q(r_{\tau})}>0,
\end{equation}
where we have used part~(2) of Assumptions~\ref{Assumption_Q}.
By \cite[equation~(3.7)]{BKS2023}, $r_{\tau}$ satisfies the following asymptotics:
\begin{equation}\label{def of rtau to 0}
 r_{\tau}=\left(\frac{\tau}{\Delta Q(0)}\right)^{1/2}+\mathcal{O}(\tau),\qquad \tau\to0.
\end{equation}
Let $\tau_{\rho}$ be a solution so that
\begin{equation}\label{def of tau rho}
 \rho q'(\rho)=2\tau_{\rho}.
\end{equation}
To split the logarithmic sum of \eqref{def of mathcal En}, we define critical indices
\begin{gather}\label{def of critical indices}
 g_{1,-}:=\lceil n(\tau_{\rho}-\delta_n) \rceil,\qquad\!\!\!
 g_{1,+}:=\lfloor n(\tau_{\rho}+\delta_n) \rfloor, \qquad\!\!\!
\delta_n':=\frac{M}{\sqrt{n}},\qquad\!\!\! M:=n^{1/8}(\log n)^{-1/8},\!\!\!
\end{gather}
where $\lceil x\rceil$ denotes the smallest integer $\geq x$, and $\lfloor x\rfloor$ denotes the largest integer $\leq x$.
We also define
\begin{equation*}
 m_n:=\bigl\lfloor n^{1/6}\bigr\rfloor,\qquad\delta_n:=\frac{\log n}{\sqrt{n}}.
\end{equation*}
See Remark~\ref{remark:choice M} for an explanation of the choice of the parameter $M$ in \eqref{def of critical indices}. See also \cite[Lemma~3.1]{BKS2023} for a similar technical choice of the parameter $m_n$ in a different context.
We write
\begin{align}\label{def of hnj in}
&h_{n,j}^{(\mathrm{in})}(\rho):=
\int_0^{\rho} 2v{\rm e}^{\mathsf{k}(v)}{\rm e}^{-nV_{\tau}(v)}|\rho-v|^{a}\,{\rm d}v,
\\
&h_{n,j}^{(\mathrm{out})}(\rho):=\int_{\rho}^{+\infty} 2v{\rm e}^{\mathsf{k}(v)}{\rm e}^{-nV_{\tau}(v)}|v-\rho|^{a}\,{\rm d}v,
\nonumber\\
&h_{n,j}:=\int_0^{+\infty} 2v{\rm e}^{\mathsf{k}(v)}{\rm e}^{-nV_{\tau}(v)}\,{\rm d}v,\nonumber
\end{align}
where $\mathsf{k}(v):=2\alpha\log v$ for $\alpha>-1$.
Then, we can rewrite $\mathcal{E}_{n}$ as
\begin{equation}
\label{def of log calEn splitted}
 \log\mathcal{E}_n=\sum_{j=0}^{n-1}\log\left[
{\rm e}^{u}\frac{h_{n,j}^{(\mathrm{in})}(\rho)}{h_{n,j}}+\frac{h_{n,j}^{(\mathrm{out})}(\rho)}{h_{n,j}}
 \right]=S_0+S_1+S_2+S_3,
\end{equation}
where
\begin{alignat*}{3}
&S_0:=\sum_{j=0}^{m_n-1}\log\left[
{\rm e}^{u}\frac{h_{n,j}^{(\mathrm{in})}(\rho)}{h_{n,j}}+\frac{h_{n,j}^{(\mathrm{out})}(\rho)}{h_{n,j}}
 \right],
\qquad&&
S_1:=\sum_{j=m_n}^{g_{1,-}-1}\log\left[
{\rm e}^{u}\frac{h_{n,j}^{(\mathrm{in})}(\rho)}{h_{n,j}}+\frac{h_{n,j}^{(\mathrm{out})}(\rho)}{h_{n,j}}
 \right],&
\\
&S_2:=\sum_{j=g_{1,-}}^{g_{1,+}}\log\left[
{\rm e}^{u}\frac{h_{n,j}^{(\mathrm{in})}(\rho)}{h_{n,j}}+\frac{h_{n,j}^{(\mathrm{out})}(\rho)}{h_{n,j}}
 \right],
\qquad&&
S_3:=\sum_{j=g_{1,+}+1}^{n-1}\log\left[
{\rm e}^{u}\frac{h_{n,j}^{(\mathrm{in})}(\rho)}{h_{n,j}}+\frac{h_{n,j}^{(\mathrm{out})}(\rho)}{h_{n,j}}
 \right].&
\end{alignat*}
For the later purpose, we split $S_2$ into
\begin{equation*}
 S_2=S_2^{(\mathrm{in})}+S_2^{(\mathrm{out})},
\end{equation*}
where
\begin{gather*}
S_2^{(\mathrm{in})}:=\sum_{j=g_{1,-}}^{\lfloor n\tau_{\rho}\rfloor-1}\log\left[
 {\rm e}^{u}\frac{h_{n,j}^{(\mathrm{in})}(\rho)}{h_{n,j}}+\frac{h_{n,j}^{(\mathrm{out})}(\rho)}{h_{n,j}}
 \right],
\\
S_2^{(\mathrm{out})}:=\sum_{j=\lfloor n\tau_{\rho}\rfloor}^{g_{1,+}}\log\left[
{\rm e}^{u}\frac{h_{n,j}^{(\mathrm{in})}(\rho)}{h_{n,j}}+\frac{h_{n,j}^{(\mathrm{out})}(\rho)}{h_{n,j}}
 \right].
\end{gather*}

\begin{Lemma}
\label{lemma:0 j Dn minus1}
For $0\leq j\leq m_n-1$, there exists $c>0$ independent of $n$ such that as $n\to+\infty$, we have
\begin{align}
\label{def of 0 j Dn in}
&h_{n,j}^{(\mathrm{in})}(\rho)=
{\rm e}^{-nq(0)}\left(\frac{2}{nq''(0)}\right)^{j+\alpha+1}
\rho^{a}\Gamma(j+\alpha+1)
\left[1+\mathcal{O}\left(\frac{(j+1)^{3/2}(\log n)^3}{\sqrt{n}}\right)\right],
\\
\label{def of 0 j Dn out}
&h_{n,j}^{(\mathrm{out})}(\rho)={\rm e}^{-nq(0)}\cdot \mathcal{O}({\rm e}^{-cn}),
\end{align}
uniformly for $a$ in compact subsets of $(-1,+\infty)$.
\end{Lemma}

\begin{proof}
We recall \cite[Lemma 4.1]{ACC2023c}
\begin{equation*}
h_{n,j}
={\rm e}^{-nq(0)}\left(\frac{2}{nq''(0)}\right)^{j+\alpha+1}
\Gamma(j+\alpha+1)
\left[1+\mathcal{O}\left(\frac{(j+1)^{3/2}(\log n)^3}{\sqrt{n}}\right)\right].
\end{equation*}
By a similar manner to the proof of \cite[Lemma~4.1]{ACC2023c}, as $n\to+\infty$, we obtain~\eqref{def of 0 j Dn in} and~\eqref{def of 0 j Dn out}.
In particular, the error terms do not depend on~$a$.
This completes the proof.
\end{proof}

\begin{Lemma}\label{lemma:S0}
For any fixed $x\in\R$, there exists $\delta>0$ such that as $n\to+\infty$, we have
\begin{equation}\label{def of S0 asymptotics}
 S_0=m_n\log(\rho^a {\rm e}^{u})+\mathcal{O}\bigl((\log n)^3n^{-1/12}\bigr),
\end{equation}
uniformly for $u\in\{z\in\C\mid |z-x|\leq \delta\}$ and $a$ in compact subsets of $(-1,+\infty)$.
\end{Lemma}
\begin{proof}By Lemma~\ref{lemma:0 j Dn minus1} and \cite[Lemma 4.1]{ACC2023c},
there exists $\delta>0$ such that we obtain \eqref{def of S0 asymptotics}, uniformly for $u\in\{z\in\C\mid |z-x|\leq \delta\}$ and $a$ in compact subsets of $(-1,+\infty)$.
\end{proof}

We next turn to $S_1$ and $S_3$.
We use the following lemma from \cite[Lemma~4.3]{ACC2023c}.
\begin{Lemma}\label{lemma:hnj asymptotics}
For $m_n\leq j\leq n-1$ and $\mathsf{k}(r)=2\alpha\log r$, as $n\to+\infty$, we have
\begin{equation*}
h_{n,j}
=\sqrt{\frac{2\pi}{n}} \frac{r_{\tau}}{\sqrt{\Delta Q(r_{\tau})}}{\rm e}^{\mathsf{k}(r_{\tau})}{\rm e}^{-nV_{\tau}(r_{\tau})}\cdot\left(1+\frac{\mathcal{A}(r_{\tau})}{n}+\mathcal{O}\left(\frac{(\log n)^{\nu}}{j^{3/2}}\right)\right)
\end{equation*}
for some $\nu>0$, where
\begin{align*}
&\mathcal{A}(r_{\tau}):=\mathcal{B}(r_{\tau})+\frac{\mathsf{k}'(r_{\tau})^2}{2}\frac{1}{d_2}+\frac{\mathsf{k}''(r_{\tau})}{2}\frac{1}{d_2}+\frac{\mathsf{k}'(r_{\tau})}{r_{\tau}}\frac{1}{d_2}-\frac{\mathsf{k}'(r_{\tau})}{2}\frac{d_3}{d_2^2},
\\
&\mathcal{B}(r):=-\frac{1}{32}\frac{\partial_r^2\Delta Q(r)}{(\Delta Q(r))^2}- \frac{19}{96}\frac{\partial_r\Delta Q(r)}{(\Delta Q(r))^2} +\frac{5}{96}\frac{(\partial_r \Delta Q(r))^2}{(\Delta Q(r))^3}+\frac{1}{12}\frac{1}{r^2\Delta Q(r)}, \\
&d_m:=V_{\tau}^{(m)}(r_{\tau}).
\end{align*}
Here, $\mathcal{O}\bigl(j^{-3/2}(\log n)^{\nu}\bigr)$ can be replaced with $\mathcal{O}\bigl(n^{-2}\bigr)$ for large $j$, i.e., for $j\geq c_0 n$ with some $c_0>0$.
\end{Lemma}
We first establish the large $n$-asymptotics of $S_1$.
To this end, for a sufficiently small $\epsilon>0$, let
$j_{1,-}:=\lceil n(\tau_{\rho}-\epsilon)\rceil$.
We write
\begin{align*}
&\theta_{1,-}^{(\epsilon)}:=j_{1,-}-n(\tau_{\rho}-\epsilon),
\qquad
\theta_{m_n}:=n^{1/6}-m_n,
\\
&\theta_{1,-}:=\lceil n(\tau_{\rho}-\delta_n)\rceil-n(\tau_{\rho}-\delta_n),\qquad
\theta_{1,+}:=n(\tau_{\rho}+\delta_n)-\lfloor n(\tau_{\rho}+\delta_n)\rfloor.
\end{align*}
Let us denote
\begin{align}\label{def of Delta nplus}
& \Delta_{n,+}:=\sqrt{\frac{n}{\Delta Q(\rho)}}\frac{\delta_n'-\tfrac{\theta_{1,+}}{n}}{\rho}=\frac{1}{\rho\sqrt{\Delta Q(\rho)}}\left(M-\frac{\theta_{1,+}}{\sqrt{n}}\right),
 \\
& \Delta_{n,-}:=\sqrt{\frac{n}{\Delta Q(\rho)}}\frac{\delta_n'-\tfrac{\theta_{1,-}}{n}}{\rho}=\frac{1}{\rho\sqrt{\Delta Q(\rho)}}\left(M-\frac{\theta_{1,-}}{\sqrt{n}}\right).\nonumber
\end{align}
By the implicit function theorem, Assumptions~\ref{Assumption_Q}, and \eqref{def of rtau ODE}, as $n\to+\infty$, $r_{\tau(g_{\pm})}$ are expanded~as
\begin{align}
\label{def of rtau g1plus}
r_{\tau(g_{1,+})}
&=
\rho+\frac{1}{2\sqrt{\Delta Q(\rho)}}\frac{\Delta_{n,+}}{\sqrt{n}}
-\frac{\Delta Q(\rho)+\rho\partial_r\Delta Q(\rho)}{8\rho \Delta Q(\rho)^2}\frac{\Delta_{n,+}^2}{n}+\mathcal{O}\left(\frac{\Delta_{n,+}^3}{n^{3/2}}\right),
\\
\label{def of rtau g1minus}
r_{\tau(g_{1,-})}
&=\rho-\frac{1}{2\sqrt{\Delta Q(\rho)}}\frac{\Delta_{n,-}}{\sqrt{n}}
-\frac{\Delta Q(\rho)+\rho\partial_r\Delta Q(\rho)}{8\rho \Delta Q(\rho)^2}\frac{\Delta_{n,-}^2}{n}+\mathcal{O}\left(\frac{\Delta_{n,-}^3}{n^{3/2}}\right).
\end{align}
We first establish the large $n$-expansion of $S_1$.
\begin{Lemma}\label{lemma:S1}
For any fixed $x\in\R$, there exists $\delta>0$ such that as $n\to+\infty$, we have
\begin{equation}\label{def of S1 asymptotics}
S_1= C_1^{(1)}n + C_2^{(1)}\sqrt{n}+C_3^{(1)}+C_n^{(1)}+
\mathcal{O}\left(\frac{(\log n)^{3}}{n^{1/12}}\right),
\end{equation}
uniformly for $u\in\{z\in\C\mid |z-x|\leq \delta\}$ and $a$ in compact subsets of $(-1,+\infty)$, where
\begin{align*}
&C_1^{(1)}:=
\tau_{\rho}u
+
\int_{0}^{\rho}2a u\Delta Q(u)\log(\rho-u)\,{\rm d}u
\qquad
C_2^{(1)}:=0,
\\
&C_3^{(1)}:=
\frac{a}{2}\log 2
+\frac{a}{4}\log\Delta Q(\rho)
+\frac{a}{2}\log\rho
-
\frac{a(a-1)}{8}\left(3+\frac{\rho \partial_r\Delta Q(\rho)}{\Delta Q(\rho)}\right)
\\
&\hphantom{C_3^{(1)}:=}{}
+
\frac{a}{4}\int_0^{\rho}\left(\frac{1}{\rho-u}
\frac{u\partial_r\Delta Q(u)}{\Delta Q(u)}
-
\frac{1}{\rho-u}
\frac{\rho\partial_r\Delta Q(\rho)}{\Delta Q(\rho)}
\right){\rm d}u
\\
& \hphantom{C_3^{(1)}:=}{}
-\frac{a}{4}
\log\bigl(2\rho\sqrt{\Delta Q(\rho)}\bigr)
\left(4\alpha+a+2-\frac{\rho \partial_r\Delta Q(\rho)}{\Delta Q(\rho)}\right),
\\
&C_n^{(1)}:=
-\rho\sqrt{\Delta Q(\rho)}\sqrt{n}\Delta_{n,-}u
-m_nu
-am_n\log\rho
\\
&\hphantom{C_n^{(1)}:=}{}
-\frac{a}{2}\rho\sqrt{\Delta Q(\rho)}
\bigl(2\log\Delta_{n,-}-2\log 2 -\log\Delta Q(\rho)-\log n-2\bigr)\Delta_{n,-}\sqrt{n}
\\
& \hphantom{C_n^{(1)}:=}{}
-\frac{a}{8}\left(1+\frac{\rho\partial_r\Delta Q(\rho)}{\Delta Q(\rho)}\right)\Delta_{n,-}^2
-\frac{a}{2}\log\Delta_{n,-}
+\frac{a}{4}\log n
+\frac{a(a-1)}{2}
\frac{\sqrt{n}}{\Delta_{n,-}}\rho\sqrt{\Delta Q(\rho)}
\\
& \hphantom{C_n^{(1)}:=}{}
-\frac{a(a-1)(2a-3)}{12}\rho\sqrt{\Delta Q(\rho)}\frac{\sqrt{n}}{\Delta_{n,-}^3}\\
& \hphantom{C_n^{(1)}:=}{}
+\frac{a}{4}\left(\log\Delta_{n,-}-\frac{1}{2}\log n\right)
\left( 4\alpha+a+2 -\frac{\rho\partial_r\Delta Q(\rho)}{\Delta Q(\rho)} \right).
\end{align*}
\end{Lemma}
\begin{proof}
By Assumptions~\ref{Assumption_Q} on the potential $Q$ and for $j\in\{m_n,\dots,g_{1,-}-1\}$, and a sufficiently large $n\in\mathbb{N}$, there exists a unique critical point $r_{\tau}\in(0,\rho)$, which satisfies $r_{\tau}q'(r_{\tau})=2\tau$.
This implies that there is no critical point inside $(\rho,+\infty)$.
By the integrability of the exponential term ${\rm e}^{-nq(r)}$,
we find that there exists $c>0$ such that
\begin{equation*}
 h_{n,j}^{(\mathrm{out})}(\rho)={\rm e}^{-nV_{\tau}(r_{\tau})}\cdot \mathcal{O}({\rm e}^{-cn}).
\end{equation*}
Therefore, we focus on \eqref{def of hnj in}. We split the integral
\begin{align*}
h_{n,j}^{(\mathrm{in})}(\rho)
&=
\int_{[0,\rho]\cap \{r\mid |r-r_{\tau}|<\delta_n\}}2r^{2\alpha+1}{\rm e}^{-nV_{\tau}(r)}(\rho-r)^{a}\,{\rm d}r
\\
&\quad
+
\int_{[0,\rho]\cap \{r\mid |r-r_{\tau}|\geq \delta_n\}}2r^{2\alpha+1}{\rm e}^{-nV_{\tau}(r)}(\rho-r)^{a}\,{\rm d}r
\\
&=
\int_{[0,\rho]\cap \{r\mid |r-r_{\tau}|<\delta_n\}}2r {\rm e}^{\mathsf{k}(r)}{\rm e}^{-nV_{\tau}(r)}(\rho-r)^{a}\,{\rm d}r
+
{\rm e}^{-nV_{\tau}(r_{\tau})}\cdot \mathcal{O}\bigl({\rm e}^{-c(\log n)^2}\bigr),
\end{align*}
where $\mathsf{k}(r):=2\alpha\log r$.
Here, the exponential small error follows from a similar manner to \cite[proof of Lemma~2.1]{BKS2023}.
For the first term in the last line, we apply the Laplace method, and then we have
\begin{align*}
&
\int_{[0,\rho]\cap \{r \mid |r-r_{\tau}|<\delta_n\}}2r {\rm e}^{\mathsf{k}(r)}{\rm e}^{-nV_{\tau}(r)}|r-\rho|^{a}\,{\rm d}r
\\
&\qquad =
\frac{2r_{\tau}{\rm e}^{\mathsf{k}(r_{\tau})}{\rm e}^{-nV_{\tau}(r_{\tau})}}{\sqrt{nd_2}}\int_{-\sqrt{d_2n}\delta_n}^{\sqrt{d_2n}\delta_n}
{\rm e}^{-\frac{1}{2}u^2}
\\
&\qquad \quad\times
\left|r_{\tau}-\rho+\frac{u}{\sqrt{nd_2}}\right|^{a}
\left(
1+\frac{c_1(u)}{\sqrt{n}}+\frac{c_2(u)}{n}+\frac{c_3(u)}{n^{3/2}}+\mathcal{O}\left(\frac{c_4(u)}{n^2}\right)
\right){\rm d}u
\\
&\qquad=
\frac{\sqrt{2\pi}r_{\tau}{\rm e}^{\mathsf{k}(r_{\tau})}{\rm e}^{-nV_{\tau}(r_{\tau})}}{\sqrt{n\Delta Q(r_{\tau})}}
(\rho-r_{\tau})^a
\biggl[
1+\frac{1}{n}\bigl(\mathcal{M}_1^{(\rm in)}(r_{\tau})+\mathcal{A}(r_{\tau})\bigr)
\\
&\qquad\quad
+\frac{1}{n^2}
\left\{
\frac{a(a-1)(a-2)(a-3)u^4}{24d_2^2(\rho-r_{\tau})^4}
+\widetilde{\mathcal{A}}(r_{\tau})
+\mathcal{O}\left(\frac{1}{(\rho-r_{\tau})^3}\right)
\right\}
\biggr]
+
\mathcal{O}\bigl({\rm e}^{-c(\log n)^2}\bigr),
\end{align*}
where the error term depends on $a$, but it does not affect the order of the error.
Here, $\widetilde{\mathcal{A}}$ is a~function with respect to $r_{\tau}$ but does not have the singularity at $\rho$, and
\begin{align*}
\mathcal{M}_1^{(\rm in)}(r_{\tau})
&=
\frac{a(a-1)(\rho-r_{\tau})^{-2}}{8\Delta Q(r_{\tau})}
-\frac{a(\rho-r_{\tau})^{-1}}{4\Delta Q(r_{\tau})}
\left(
-\frac{\partial_r\Delta Q(r_{\tau})}{2\Delta Q(r_{\tau})}
+\frac{4\alpha+3}{2r_{\tau}}
\right).
\end{align*}
and, we safely extended the integral region to $(-\infty,\infty)$ with the exponential error $\mathcal{O}\bigl({\rm e}^{-c(\log n)^2}\bigr)$ for some $c>0$, and for $d_{j}:=V_{\tau}^{(j)}(r_{\tau})$ given by \eqref{def of V tau relationship 1}, $c_k(u)$ for $k=1,2,3$ are given by
\begin{align}
\label{def of c_1}
&c_1(u):=-\frac{d_3}{6d_2^{3/2}}u^3+\left(\mathsf{k}'(r_{\tau})+\frac{1}{r_{\tau}}\right)\frac{1}{d_2^{1/2}}u,
\\
\label{def of c_2}
&c_2(u):=\frac{d_3^2}{72d_2^3}u^6-\left(d_4+4d_3\mathsf{k}'(r_{\tau})+\frac{4d_3}{r_{\tau}}\right)\frac{u^4}{24d_2^2} +\left(\frac{\mathsf{k}''(r_{\tau})}{2}+\frac{\mathsf{k}'(r_{\tau})^2}{2} +\frac{\mathsf{k}'(r_{\tau})}{r_{\tau}}\right)\frac{u^2}{d_2},\!\!\!\!
\\
&c_3(u):=-\frac{d_3^3}{1296d_2^{9/2}}u^9+\left(d_3d_4+2d_3^2\mathsf{k}'(r_{\tau})+\frac{2d_3^2}{r_{\tau}} \right)\frac{u^7}{144d_2^{7/2}}\nonumber
\\
&\hphantom{c_3(u):=}{}
-\left(
d_5+\frac{5}{r_{\tau}}d_4+5d_4\mathsf{k}'(r_{\tau})+\frac{20d_3}{r_{\tau}}\mathsf{k}'(r_{\tau}) +10d_3\mathsf{k}'(r_{\tau})^2+10d_3\mathsf{k}''(r_{\tau})
\right)
\frac{u^5}{120d_2^{5/2}}
\nonumber
\\
&\hphantom{c_3(u):=}{}
+
\left(\frac{3\mathsf{k}'(r_{\tau})^2}{r_{\tau}}
+\frac{3\mathsf{k}''(r_{\tau})}{r_{\tau}}+\mathsf{k}'(r_{\tau})^3 +3\mathsf{k}'(r_{\tau})\mathsf{k}''(r_{\tau})+\mathsf{k}^{(3)}(r_{\tau})\right)\frac{u^3}{6d_2^{3/2}},\label{def of c_3}
\end{align}
and $c_4(u)$ is a polynomial consisting of $u^{12}$, $u^{10}$, $u^{8}$, $u^6$, $u^4$.
Since $\rho-r_{\tau}>0$ and $\rho-r_{\tau}\searrow c n^{-1/2}M$ with some $c>0$ for $m_n\leq j\leq g_{1,-}-1$, by Lemma~\ref{lemma:hnj asymptotics}, we have
\begin{align*}
&S_1=
\sum_{j=D_n}^{g_{1,-}-1}
\log\bigl({\rm e}^{u}(\rho-r_{\tau})^{a}\bigr)
\\
&\hphantom{S_1=}{}
+\frac{1}{n}\sum_{j=j_{1,-}}^{g_{1,-}-1}
\left[
\frac{a(a-1)(\rho-r_{\tau})^{-2}}{8\Delta Q(r_{\tau})}
-\frac{a(\rho-r_{\tau})^{-1}}{4\Delta Q(r_{\tau})}
\left(
-\frac{\partial_r\Delta Q(r_{\tau})}{2\Delta Q(r_{\tau})}+\frac{4\alpha+3}{2r_{\tau}}
\right)
\right]
\\
&\hphantom{S_1=}{}
-\frac{1}{n^2}
\sum_{j=D_n}^{g_{1,-}-1}
\frac{a(a-1)(2a-3)}{64\Delta Q(r_{\tau})^2(\rho-r_{\tau})^4}
+\sum_{j=D_n}^{g_{1,-}-1}
\mathcal{O}\left(\frac{1}{n^{2}}\frac{1}{(\rho-r_{\tau})^3}\right)
+
\mathcal{O}\left(\frac{(\log n)^{3}}{n^{1/12}}\right),
\end{align*}
where the second error term is independent of $a$, $u$.
By Lemma~\ref{lemma:Euler--Maclaurin} and by change of variable $r_{\tau(j)}=t$, we have
\[
\sum_{j=m_n}^{g_{1,-}-1}
\mathcal{O}\left(\frac{1}{n^{2}}\frac{1}{(\rho-r_{\tau})^3}\right)
=
\mathcal{O}\left(\frac{(\log n)^{3/8}}{n^{7/8}}\right).
\]
Consequently, the expansion of $S_1$ takes the following form:
\begin{align*}
&S_1=(g_{1,-}-m_n)u+a\int_{m_n}^{g_{1,-}}\log(\rho-r_{\tau(t)})\,{\rm d}t
-\frac{a}{2}\log(\rho-r_{\tau(g_{1,-})})
+\frac{a}{2}\log(\rho-r_{\tau(D_n)})
\\
&\hphantom{S_1=}{}
+\frac{1}{n}\sum_{j=D_n}^{g_{1,-}-1}
\left[
\frac{a(a-1)(\rho-r_{\tau})^{-2}}{8\Delta Q(r_{\tau})}
-\frac{a(\rho-r_{\tau})^{-1}}{4\Delta Q(r_{\tau})}
\left(
-\frac{\partial_r\Delta Q(r_{\tau})}{2\Delta Q(r_{\tau})}+\frac{4\alpha+3}{2r_{\tau}}
\right)
\right]
\\
&\hphantom{S_1=}{}
-\frac{1}{n^2}
\sum_{j=m_n}^{g_{1,-}-1}
\frac{a(a-1)(2a-3)}{64\Delta Q(r_{\tau})^2(\rho-r_{\tau})^4}
+\mathcal{O}\left(\frac{(\log n)^{3}}{n^{1/12}}\right),
\end{align*}
where we have used Lemma~\ref{lemma:Euler--Maclaurin}, and the error term is independent of $a$,~$u$.
By \eqref{def of rtau to 0} and \eqref{def of rtau g1minus}, we have
\begin{align*}
& a\int_{m_n}^{g_{1,-}}\log(\rho-r_{\tau(t)})\,{\rm d}t\\
&\qquad =
n\int_{0}^{\rho}2a u\Delta Q(u)\log(\rho-u)\,{\rm d}u
-am_n\log\rho
\\
&\qquad \quad
-\frac{a}{2}\rho\sqrt{\Delta Q(\rho)}
\bigl(2\log\Delta_{n,-}-2\log 2 -\log\Delta Q(\rho)-\log n-2\bigr)\Delta_{n,-}\sqrt{n}
\\
&\qquad \quad
-\frac{a}{8}\left(1+\frac{\rho\partial_r\Delta Q(\rho)}{\Delta Q(\rho)}\right)\Delta_{n,-}^2
+\mathcal{O}\left(\frac{(\log n)^{1/4}}{n^{1/4}}\right),
\end{align*}
and
\begin{align*}
& -\frac{a}{2}\log(\rho-r_{\tau(g_{1,-})})+\frac{a}{2}\log(\rho-r_{\tau(D_n)})
\\
&\qquad =
-\frac{a}{2}\log\Delta_{n,-}
+\frac{a}{2}\log 2
+\frac{a}{4}\log\Delta Q(\rho)
+\frac{a}{4}\log n
+\frac{a}{2}\log\rho
+\mathcal{O}\left(\frac{(\log n)^{1/4}}{n^{1/4}}\right).
\end{align*}
The Taylor theorem gives rise to
\begin{align*}
&
a\int_{m_n}^{g_{1,-}}\log(\rho-r_{\tau(t)})\,{\rm d}t
-\frac{a}{2}\log(\rho-r_{\tau(g_{1,-})})
+\frac{a}{2}\log(\rho-r_{\tau(m_n)})
\\
&\qquad =
n\int_{0}^{\rho}2a u\Delta Q(u)\log(\rho-u)\,{\rm d}u
+\frac{a}{2}\log 2
+\frac{a}{4}\log\Delta Q(\rho)
+\frac{a}{2}\log\rho
\\
&\qquad \quad
-\frac{a}{2}\rho\sqrt{\Delta Q(\rho)}
\bigl(2\log\Delta_{n,-}-2\log 2 -\log\Delta Q(\rho)-\log n-2\bigr)\Delta_{n,-}\sqrt{n}
\\
&\qquad \quad
-am_n\log\rho
-\frac{a}{8}\left(1+\frac{\rho\partial_r\Delta Q(\rho)}{\Delta Q(\rho)}\right)\Delta_{n,-}^2
-\frac{a}{2}\log\Delta_{n,-}
+\frac{a}{4}\log n
+\mathcal{O}\left(\frac{(\log n)^{1/4}}{n^{1/4}}\right).
\end{align*}
By Lemma~\ref{lemma:Euler--Maclaurin}, we have
\begin{align*}
-\frac{1}{n^2}
\sum_{j=m_n}^{g_{1,-}-1}
\frac{a(a-1)(2a-3)}{64\Delta Q(r_{\tau})^2(\rho-r_{\tau})^4}
&=
-\frac{1}{n}\int_{r_{\tau(m_n)}}^{r_{\tau(g_{1,-})}}
\frac{a(a-1)(2a-3)}{32\Delta Q(u)(\rho-u)^4}u\,{\rm d}u+\mathcal{O}\bigl(\Delta_{n,-}^{-4}\bigr).
\end{align*}
Integration by parts leads to
\begin{gather*}
-\frac{1}{n}\int_{r_{\tau(m_n)}}^{r_{\tau(g_{1,-})}}
\frac{a(a-1)(2a-3)}{32\Delta Q(u)(\rho-u)^4}u\,{\rm d}u
=
-
\frac{1}{n}\frac{a(a-1)(2a-3)r_{\tau(g_{1,-})}}{96\Delta Q(r_{\tau(g_{1,-})})(\rho-r_{\tau(g_{1,-})})^3}
+
\widetilde{\epsilon}_{n,-}
+
\mathcal{O}\bigl(n^{-1}\bigr),
\end{gather*}
where
\[
\widetilde{\epsilon}_{n,-}:=\frac{1}{n}\int_{r_{\tau(m_n)}}^{r_{\tau(g_{1,-})}}
\frac{a(a-1)(2a-3)}{96(\rho-u)^3}\frac{\Delta Q(u)-u\partial_u\Delta Q(u)}{\Delta Q(u)^2}\,{\rm d}u.
\]
Note that
\begin{gather*}
-
\frac{1}{n}\frac{a(a-1)(2a-3)r_{\tau(g_{1,-})}}{96\Delta Q(r_{\tau(g_{1,-})})(\rho-r_{\tau(g_{1,-})})^3}
=-\frac{a(a-1)(2a-3)}{12}\rho\sqrt{\Delta Q(\rho)}\frac{\sqrt{n}}{\Delta_{n,-}^3}+\mathcal{O}\bigl(\Delta_{n,-}^{-2}\bigr),
\\
|\widetilde{\epsilon}_{n,-}|\lesssim \frac{1}{n}\biggl|\int_{r_{\tau(m_n)}}^{r_{\tau(g_{1,-})}}\frac{{\rm d}u}{(\rho-u)^3}\biggr|=\mathcal{O}\bigl(\Delta_{n,-}^{-2}\bigr).
\end{gather*}
By Lemma~\ref{lemma:Euler--Maclaurin}, we have
\begin{align*}
& \frac{1}{n}\sum_{j=m_n}^{g_{1,-}-1}
\left[
\frac{a(a-1)(\rho-r_{\tau})^{-2}}{8\Delta Q(r_{\tau})}
-\frac{a(\rho-r_{\tau})^{-1}}{4\Delta Q(r_{\tau})}
\left(
-\frac{\partial_r\Delta Q(r_{\tau})}{2\Delta Q(r_{\tau})}+\frac{4\alpha+3}{2r_{\tau}}
\right)
\right]
\\
&\qquad =
\frac{a(a-1)}{4}\left[\frac{\rho}{\rho-r_{\tau(g_{1,-})}} +\log(\rho-r_{\tau(g_{1,-})})-\frac{\rho}{\rho-r_{\tau(m_n)}}-\log(\rho-r_{\tau(m_n)})\right]
\\
&\qquad \quad
+\frac{a}{4}\bigl(\log(\rho-r_{\tau(g_{1,-})})-\log(\rho-r_{\tau(m_n)})\bigr)
\left(4\alpha+3-\frac{\rho \partial_r\Delta Q(\rho)}{\Delta Q(\rho)}\right)
\\
&\qquad \quad
+\frac{a}{4}\int_{r_{\tau(m_n)}}^{r_{\tau(g_{1,-})}}\left(\frac{1}{\rho-u}\frac{u \partial_r\Delta Q(u)}{\Delta Q(u)}-\frac{1}{\rho-u}\frac{\rho \partial_r\Delta Q(\rho)}{\Delta Q(\rho)}\right) {\rm d}u
+\mathcal{O}\bigl(\Delta_{n,-}^{-2}\bigr).
\end{align*}
Using the asymptotics by the Taylor theorem,
\begin{align*}
& \frac{a}{4}\int_{r_{\tau(m_n)}}^{r_{\tau(g_{1,-})}}\left(\frac{1}{\rho-u}\frac{u \partial_r\Delta Q(u)}{\Delta Q(u)}-\frac{1}{\rho-u}\frac{\rho \partial_r\Delta Q(\rho)}{\Delta Q(\rho)}\right){\rm d}u
\\
&\qquad =
\frac{a}{4}\int_{0}^{\rho}\left(\frac{1}{\rho-u}\frac{u \partial_r\Delta Q(u)}{\Delta Q(u)}-\frac{1}{\rho-u}\frac{\rho \partial_r\Delta Q(\rho)}{\Delta Q(\rho)}\right){\rm d}u
+\mathcal{O}\left(\frac{\Delta_{n,-}}{\sqrt{n}}+\frac{m_n}{n}\right),
\\
& \frac{\rho}{\rho-r_{\tau(g_{1,-})}}+\log(\rho-r_{\tau(g_{1,-})})-\frac{\rho}{\rho-r_{\tau(D_n)}}-\log(\rho-r_{\tau(m_n)})
\\
&\qquad =
2\rho\sqrt{\Delta Q(\rho)}\frac{\sqrt{n}}{\Delta_{n,-}}-\frac{1}{2}\left(3+\frac{\rho \partial_r\Delta Q(\rho)}{\Delta Q(\rho)}\right)-\log\bigl(2\rho\sqrt{\Delta Q(\rho)}\bigr)\\
&\qquad \quad -\frac{1}{2}\log n+\log\Delta_{n,-}+\mathcal{O}\left(\sqrt{\frac{m_n}{n}}\right),
\end{align*}
and
\begin{gather*}
\log(\rho-r_{\tau(g_{1,-})})-\log(\rho-r_{\tau(D_n)})
=-\frac{1}{2}\log n+\log\Delta_{n,-}\!-\log\bigl(2\rho\sqrt{\Delta Q(\rho)}\bigr)\!
+\mathcal{O}\!\left(\!\sqrt{\frac{m_n}{n}}\right),
\end{gather*}
we have
\begin{align*}
& \frac{1}{n}\sum_{j=m_n}^{g_{1,-}-1}
\left[
\frac{a(a-1)(\rho-r_{\tau})^{-2}}{8\Delta Q(r_{\tau})}
-\frac{a(\rho-r_{\tau})^{-1}}{4\Delta Q(r_{\tau})}
\left(
-\frac{\partial_r\Delta Q(r_{\tau})}{2\Delta Q(r_{\tau})}+\frac{4\alpha+3}{2r_{\tau}}
\right)
\right]
\\
&\qquad =
-
\frac{a(a-1)}{4}\left[
\frac{1}{2}\left(3+\frac{\rho \partial_r\Delta Q(\rho)}{\Delta Q(\rho)}\right)
+\log\bigl(2\rho\sqrt{\Delta Q(\rho)}\bigr)
\right]
\\
&\qquad \quad
-\frac{a}{4}
\log\bigl(2\rho\sqrt{\Delta Q(\rho)}\bigr)
\left(4\alpha+3-\frac{\rho \partial_r\Delta Q(\rho)}{\Delta Q(\rho)}\right)
\\
&\qquad \quad
+\frac{a}{4}\int_{0}^{\rho}\left(\frac{1}{\rho-u}\frac{u \partial_r\Delta Q(u)}{\Delta Q(u)}-\frac{1}{\rho-u}\frac{\rho \partial_r\Delta Q(\rho)}{\Delta Q(\rho)}\right){\rm d}u
\\
&\qquad \quad
+\frac{a}{4}\left(
-\frac{1}{2}\log n+\log\Delta_{n,-}
\right)
\left(4\alpha+a+2-\frac{\rho \partial_r\Delta Q(\rho)}{\Delta Q(\rho)}\right)
\\
&\qquad \quad
+\frac{a(a-1)}{2}
\rho\sqrt{\Delta Q(\rho)}\frac{\sqrt{n}}{\Delta_{n,-}}
+\mathcal{O}\bigl(\Delta_{n,-}^{-2}\bigr).
\end{align*}
Combining all the above, we obtain \eqref{def of S1 asymptotics}.
\end{proof}

Next, we compute the large $n$-asymptotics of $S_3$.
\begin{Lemma}\label{lemma:S3}
For any fixed $x\in\R$, there exists $\delta>0$ such that as $n\to+\infty$, we have
\begin{equation}
\label{def of S3 asymptotics}
 S_3=
 C_1^{(3)}n+C_2^{(3)}\sqrt{n}+C_3^{(3)}+C_n^{(3)}+\mathcal{O}\left(\frac{(\log n)^{1/4}}{n^{1/4}}\right),
\end{equation}
uniformly for $u\in\{z\in\C\mid |z-x|\leq \delta\}$ and $a$ in compact subsets of $(-1,+\infty)$, where
\begin{align*}
&C_1^{(3)}:=
a\int_{\rho}^{r_1}\log (t-\rho) 2t\Delta Q(t)\,{\rm d}t,
\qquad
C_2^{(3)}:=0,
\\
&C_3^{(3)}:=
-\frac{a}{2}\log(r_1-\rho)
+\frac{a}{2}\log2
+\frac{a}{4}\log(\Delta Q(\rho))\\
& \hphantom{C_3^{(3)}:=}{}
+\frac{a(a-1)}{4}\left[\frac{1}{2}\left(1+\frac{\rho \partial_r\Delta Q(\rho)}{\Delta Q(\rho)}\right)-\frac{\rho}{r_1-\rho}\right]
\\
&\hphantom{C_3^{(3)}:=}{}
+\frac{a}{4}\left(4\alpha+a+2-\frac{\rho \partial_r\Delta Q(\rho)}{\Delta Q(\rho)}\right)\log\bigl(2(r_1-\rho)\sqrt{\Delta Q(\rho)}\bigr)
\\
&\hphantom{C_3^{(3)}:=}{}
-\frac{a}{4}\int_{\rho}^{r_1}\left(\frac{1}{t-\rho}\frac{t \partial_t\Delta Q(t)}{\Delta Q(t)}-\frac{1}{t-\rho}\frac{\rho \partial_r\Delta Q(\rho)}{\Delta Q(\rho)}\right) {\rm d}t,
\\
&C_n^{(3)}:=
-\frac{a}{2}\rho\sqrt{\Delta Q(\rho)}
\left(
2\log\Delta_{n,+}-2\log 2-\log\Delta Q(\rho)-\log n-2
\right)\sqrt{n}\Delta_{n,+}
\\
&\hphantom{C_n^{(3)}:=}{}
+\frac{a}{8}\left(1+\frac{\rho\partial_r\Delta Q(\rho)}{\Delta Q(\rho)}\right)\Delta_{n,+}^2
-\frac{a}{2}\log(\Delta_{n,+})
+\frac{a}{4}\log n
+\frac{a(a-1)}{2}\rho\sqrt{\Delta Q(\rho)}\frac{\sqrt{n}}{\Delta_{n,+}}
\\
& \hphantom{C_n^{(3)}:=}{}
-\frac{a}{4}
\left(
\log\Delta_{n,+}-\frac{1}{2}\log n
\right)
\left(
4\alpha+a+2-\frac{\rho\partial_r\Delta Q(\rho)}{\Delta Q(\rho)}
\right)\\
& \hphantom{C_n^{(3)}:=}{}
-\frac{a(a-1)(2a-3)}{12}\rho\sqrt{\Delta Q(\rho)}\frac{\sqrt{n}}{\Delta_{n,+}^3}.
\end{align*}
\end{Lemma}

\begin{proof}
Similar manner to Lemma~\ref{lemma:S1}, there exists $c>0$ such that as $n\to+\infty$, we have
\begin{align}
&h_{n,j}^{(\mathrm{out})}(\rho)
=
\frac{\sqrt{2\pi}r_{\tau}{\rm e}^{\mathsf{k}(r_{\tau})}{\rm e}^{-nV_{\tau}(r_{\tau})}}{\sqrt{n\Delta Q(r_{\tau})}}
(r_{\tau}-\rho)^a
\biggl[
1+\frac{1}{n}
\left(
\mathcal{A}_1(r_{\tau})
+\mathcal{M}_1^{(\mathrm{out})}(r_{\tau})
\right)\nonumber
\\
&\hphantom{h_{n,j}^{(\mathrm{out})}(\rho)=}{}
+\frac{1}{n^2}
\left\{
\frac{a(a-1)(a-2)(a-3)u^4}{24d_2^2(r_{\tau}-\rho)^4}
+\mathcal{A}_2(r_{\tau})
+\mathcal{O}\left(\frac{1}{(r_{\tau}-\rho)^3}\right)
\right\}
\biggr]\nonumber\\
&\hphantom{h_{n,j}^{(\mathrm{out})}(\rho)=}{}
+
\mathcal{O}\bigl({\rm e}^{-c(\log n)^2}\bigr),\label{def of hnj out in S3}
\end{align}
where the constant in the error term might depend on $a$, but the order of the error term is independent of $a$. Here,
\[
\mathcal{M}_1^{(\mathrm{out})}(r_{\tau})
=
\frac{a(a-1)(r_{\tau}-\rho)^{-2}}{8\Delta Q(r_{\tau})}
+\frac{a(r_{\tau}-\rho)^{-1}}{8\Delta Q(r_{\tau})}
\left(-\frac{\partial_r\Delta Q(r_{\tau})}{\Delta Q(r_{\tau})}
+\frac{4\alpha+3}{r_{\tau}}
\right).
\]
Also, there exists $c>0$ independent of $n$ such that we have
\begin{equation}\label{def of hnj in in S3}
{\rm e}^{u}\frac{h_{n,j}^{(\mathrm{in})}(\rho)}{h_{n,j}}={\rm e}^{-nV_{\tau}(r_{\tau})}\cdot\mathcal{O}({\rm e}^{-cn}),
\end{equation}
where we can take the error term to be independent of $u$.
Thus, by Lemma~\ref{lemma:hnj asymptotics}, \eqref{def of hnj out in S3}, and \eqref{def of hnj in in S3}, let us choose $\delta>0$ sufficiently small so that
\[
{\rm e}^{u}\frac{h_{n,j}^{(\mathrm{in})}(\rho)}{h_{n,j}}+\frac{h_{n,j}^{(\mathrm{out})}(\rho)}{h_{n,j}}
\]
remains bounded away from the interval $(-\infty,0]$ as $n\to+\infty$ uniformly for $u\in\{z\in\C \mid \allowbreak |z-x|\leq \delta\}$ and $a$ in compact subsets of $(-1,+\infty)$.
Therefore, as $n\to+\infty$, we have
\begin{align*}
& \sum_{j=g_{1,+}+1}^{n-1}\log\left[
{\rm e}^{u}\frac{h_{n,j}^{(\mathrm{in})}(\rho)}{h_{n,j}}+\frac{h_{n,j}^{(\mathrm{out})}(\rho)}{h_{n,j}}
 \right]
=
a\sum_{j=g_{1,+}+1}^{n-1}\log(r_{\tau}-\rho)\\
& \qquad {} +
\frac{1}{n}
\sum_{j=g_{1,+}+1}^{n-1}
\left\{
\frac{a(a-1)(r_{\tau}-\rho)^{-2}}{8\Delta Q(r_{\tau})}
+\frac{a(r_{\tau}-\rho)^{-1}}{8\Delta Q(r_{\tau})}
\left(-\frac{\partial_r\Delta Q(r_{\tau})}{\Delta Q(r_{\tau})}
+\frac{4\alpha+3}{r_{\tau}}
\right)
\right\}
\\
&\qquad{}
-\frac{1}{n^2}
\sum_{j=g_{1,+}+1}^{n-1}
\frac{a(a-1)(2a-3)}{64\Delta Q(r_{\tau})^2(r_{\tau}-\rho)^4}
+
\sum_{j=g_{1,+}+1}^{n-1}
\mathcal{O}\left(\frac{(r_{\tau}-\rho)^{-3}}{n^2}\right),
\end{align*}
uniformly for $u\in\{z\in\C\mid |z-x|\leq \delta\}$ and $a$ in compact subsets of $(-1,+\infty)$.
By change of variable $r_{\tau(j)}=t$ and \eqref{def of rtau ODE}, we have
\begin{equation}\label{def of summation in 4 S3}
\sum_{j=g_{1,+}+1}^{n-1}
\mathcal{O}\left(\frac{(r_{\tau}-\rho)^{-3}}{n^2}\right)
=
\mathcal{O}\left(\frac{(\log n)^{3/8}}{n^{7/8}}\right).
\end{equation}
By Lemma~\ref{lemma:Euler--Maclaurin}, we have
\begin{align*}
&\sum_{j=g_{1,+}+1}^{n-1}a\log(r_{\tau}-\rho)\\
&\qquad =
\int_{g_{1,+}}^{n}a\log(r_{\tau(t)}-\rho)\,{\rm d}t
-\frac{a}{2}\log(r_1-\rho)-\frac{a}{2}\log(r_{\tau(g_{1,+})}-\rho)
+\mathcal{O}\bigl(n^{-3/8}\bigr).
\end{align*}
Note that by \eqref{def of rtau g1plus}, as $n\to+\infty$, we have
\begin{align}
\nonumber
&\int_{g_{1,+}}^{n}a\log(r_{\tau(t)}-\rho)\,{\rm d}t
=n a\int_{\rho}^{r_1}\log (u-\rho) 2u\Delta Q(u)\,{\rm d}u
\\
&\qquad{}
-a\rho\sqrt{\Delta Q(\rho)}
\left(
\log\Delta_{n,+}-\log 2-\frac{1}{2}\log\Delta Q(\rho)-\frac{1}{2}\log n-1
\right)\sqrt{n}\Delta_{n,+}
\nonumber
\\
&\qquad{}
+\frac{a}{8}\left(1+\frac{\rho\partial_r\Delta Q(\rho)}{\Delta Q(\rho)}\right)\Delta_{n,+}^2
+\mathcal{O}\left(\frac{\Delta_{n,+}^3}{\sqrt{n}}\right).\label{def of choice M reason}
\end{align}
Thus, the Taylor theorem together with \eqref{def of rtau g1plus} gives rise to
\begin{align}
&\sum_{j=g_{1,+}+1}^{n-1}a\log(r_{\tau}-\rho)
=
n a\int_{r_{\tau}(g_{1,+})}^{r_1}\log (u-\rho) 2u\Delta Q(u)\,{\rm d}u\nonumber
\\
&\qquad{}
+\frac{a}{2}\rho\sqrt{\Delta Q(\rho)}
\left(
2\log 2+2+\log\Delta Q(\rho)+\log n
-2\log\Delta_{n,+}
\right)\sqrt{n}\Delta_{n,+}\nonumber
\\
&\qquad{}
+\frac{a}{8}\left(1+\frac{\rho\partial_r\Delta Q(\rho)}{\Delta Q(\rho)}\right)\Delta_{n,+}^2
-\frac{a}{2}\log(r_1-\rho)
-\frac{a}{2}\log(\Delta_{n,+})
+\frac{a}{2}\log2\nonumber
\\
&\qquad{}
+\frac{a}{4}\log(\Delta Q(\rho))
+\frac{a}{4}\log n
+
\mathcal{O}\left(
\frac{1}{\Delta_{n,+}}
\right).\label{def of summation in 1 S3}
\end{align}
The term of order $n^{-1}$ can be calculated as
\begin{align}
&
\frac{1}{n}
\sum_{j=g_{1,+}+1}^{n-1}
\left\{
\frac{a(a-1)(r_{\tau}-\rho)^{-2}}{8\Delta Q(r_{\tau})}
+\frac{a(r_{\tau}-\rho)^{-1}}{8\Delta Q(r_{\tau})}
\left(-\frac{\partial_r\Delta Q(r_{\tau})}{\Delta Q(r_{\tau})}
+\frac{4\alpha+3}{r_{\tau}}
\right)
\right\}
\nonumber\\
&\qquad {}=
\frac{a(a-1)}{4}\left[\frac{1}{2}\left(1+\frac{\rho \partial_r\Delta Q(\rho)}{\Delta Q(\rho)}\right)+\log(2(r_1-\rho)\sqrt{\Delta Q(\rho)})-\frac{\rho}{r_1-\rho}\right]
\nonumber\\
&\qquad\quad{}
+\frac{a}{4}\left(4\alpha+3-\frac{\rho \partial_r\Delta Q(\rho)}{\Delta Q(\rho)}\right)\log\bigl(2(r_1-\rho)\sqrt{\Delta Q(\rho)}\bigr)
\nonumber\\
&\qquad\quad{}
-\frac{a}{4}\int_{\rho}^{r_1}\left(\frac{1}{u-\rho}\frac{u \partial_u\Delta Q(u)}{\Delta Q(u)}-\frac{1}{u-\rho}\frac{\rho \partial_r\Delta Q(\rho)}{\Delta Q(\rho)}\right){\rm d}u
+\frac{a(a-1)}{2}\rho\sqrt{\Delta Q(\rho)}\frac{\sqrt{n}}{\Delta_{n,+}}
\nonumber\\
&\qquad\quad{}
+\frac{a}{4}\left(\frac{1}{2}\log n-\log\Delta_{n,+}\right)\left(4\alpha+a+2-\frac{\rho \partial_r\Delta Q(\rho)}{\Delta Q(\rho)}\right)+\mathcal{O}\bigl(\Delta_{n,+}^{-2}\bigr).
\label{def of summation in 2 S3}
\end{align}
Finally, it is straightforward to see that by \eqref{def of rtau g1plus} and Lemma~\ref{lemma:Euler--Maclaurin},
\begin{align}
&-\frac{1}{n^2}
\sum_{j=g_{1,+}+1}^{n-1}
\frac{a(a-1)(2a-3)}{64\Delta Q(r_{\tau})^2(r_{\tau}-\rho)^4}
\nonumber\\
&\qquad =
-\frac{a(a-1)(2a-3)}{12}\rho\sqrt{\Delta Q(\rho)}\frac{\sqrt{n}}{\Delta_{n,+}^3}
+\mathcal{O}\left(\frac{(\log n)^{1/4}}{n^{1/4}}\right).\label{def of summation in 3 S3}
\end{align}
Combining \eqref{def of summation in 4 S3} with \eqref{def of summation in 1 S3}, \eqref{def of summation in 2 S3}, and \eqref{def of summation in 3 S3}, we obtain \eqref{def of S3 asymptotics}.
\end{proof}

\begin{Remark}\label{remark:choice M}
The reader might wonder whether the choice of the parameter $M$ in \eqref{def of critical indices} is merely technical. Indeed, as shown in Lemma~\ref{lemma:S3}, the constant \smash{$C_n^{(3)}$} contains the divergent term $ \sqrt{n}/\Delta_{n,+}^3
=\mathcal{O}\bigl(n^{1/8}(\log n)^{3/8}\bigr)$ as $n\to+\infty$ (see~\eqref{def of Delta nplus} for the definition of $\Delta_{n,+}$). On~the other hand, the proof of Lemma~\ref{lemma:S3} also produces a term of order $\Delta_{n,+}^3/\sqrt{n}$ in~\eqref{def of choice M reason}.
Although the divergent contribution $\sqrt{n}/\Delta_{n,+}^3$ is eventually cancelled in the final step of the proof (see Section~\ref{subsection:last proof}), the remaining term $\Delta_{n,+}^3/\sqrt{n}$ must still be absorbed into the error term. The parameter~$M$ in \eqref{def of critical indices} is therefore chosen precisely so that both of these requirements are satisfied simultaneously. A~similar technical choice of parameters also appeared in a different context in \cite[Remark~3.17]{C2021}.
\end{Remark}

\section{Local asymptotic analysis for the proof of Theorem~\ref{theorem:calEn expansion}}
\label{section:local analysis part}
\subsection{Asymptotic expansion of \texorpdfstring{$\boldsymbol{S_2}$}{S\_2}}
\label{subsection:local analysis S2}
We begin with the asymptotic expansion of the summand \eqref{def of log calEn splitted}.
We recall \eqref{def of parabolic cylinder function} here.
\begin{Lemma}\label{lemma:S2in part1}
Let
\begin{equation}
\label{def of xi}
 \xi\equiv\xi_{j}:=\frac{\sqrt{n}}{\sqrt{\Delta Q(\rho)}}\frac{\tau-\tau_{\rho}}{\rho},
\end{equation}
where $\tau_{\rho}$ is given by \eqref{def of tau rho}.
There exists $\delta>0$ such that as $n\to+\infty$, we have
\begin{align}
&S_2 =
\sum_{j=g_{1,-}}^{g_{1,+}}\log\left(\frac{1}{(4n)^{a/2}}\right)
+
\frac{a}{2}\sum_{j=g_{1,-}}^{g_{1,+}}\log\left(\frac{1}{\Delta Q(\rho)}\right)
+
\sum_{j=g_{1,-}}^{g_{1,+}}\log [\mathcal{H}_{a,u}(\xi) ]
\nonumber\\
&\hphantom{S_2 =}{}
+\frac{1}{\sqrt{n}} \sum_{j=g_{1,-}}^{g_{1,+}}\frac{\widetilde{g}_{a,u}(\xi)}{12\rho\sqrt{\Delta Q(\rho)}g_{a,u}(\xi)}
+\mathcal{O}\left(\frac{(\log n)^{7/8}}{n^{1/8}}\right),\label{def of S2 asymptotics part 1}
\end{align}
uniformly for $u\in\{z\in\C\mid |z-x|\leq \delta\}$ and $a$ in compact subsets of $(-1,+\infty)$.
Here, $\mathcal{H}_{a,u}(\xi),g_{a,u}(\xi)$ are given by \eqref{def of cal H au}, and
\begin{align}
&\widetilde{g}_{a,u}(\xi):=
-a\xi\left(1-\frac{\rho \partial_r\Delta Q(\rho)}{\Delta Q(\rho)}\right)
({\rm e}^{u}D_{-a-1}(\xi)+D_{-a-1}(-\xi))
\nonumber\\
&\hphantom{\widetilde{g}_{a,u}(\xi):=}{}
+({\rm e}^u D_{-a}(\xi)-D_{-a}(-\xi))
\xi^2\left(2+\frac{\rho \partial_r\Delta Q(\rho)}{\Delta Q(\rho)}\right)
\nonumber\\
&\hphantom{\widetilde{g}_{a,u}(\xi):=}{}
-({\rm e}^u D_{-a}(\xi)-D_{-a}(-\xi))
\left[
6(\rho \mathsf{k}'(\rho)+1)+(2+a)\left(1-\frac{\rho \partial_r\Delta Q(\rho)}{\Delta Q(\rho)}\right)
\right].\label{def of tilde gau}
\end{align}
\end{Lemma}
\begin{proof}
We begin with analyzing $S_2^{(\mathrm{in})}$.
There exists $c>0$ such that
\begin{align}
\nonumber
&h_{n,j}^{(\mathrm{in})}(\rho)
=
\int_{r_{\tau}-\delta_n'}^{\rho}
2r {\rm e}^{\mathsf{k}(r)}{\rm e}^{-nV_{\tau}(r)}|\rho-r|^{a}\,{\rm d}r
+\int_{0}^{r_{\tau}-\delta_n'}
2r {\rm e}^{\mathsf{k}(r)}{\rm e}^{-nV_{\tau}(r)}|\rho-r|^{a}\,{\rm d}r
\\
&\hphantom{h_{n,j}^{(\mathrm{in})}(\rho)}{} =
\label{def of hin rho 1}
\frac{2r_{\tau}{\rm e}^{\mathsf{k}(r_{\tau})}{\rm e}^{-nV_{\tau}(r_{\tau})}}{(nd_2)^{(a+1)/2}}
\int_{0}^{\sqrt{d_2n}(\rho-r_{\tau}+\delta_n')}
v^{a}
{\rm e}^{-\frac{1}{2}(v-\sqrt{nd_2}(\rho-r_{\tau}))^2}
\\
&\hphantom{h_{n,j}^{(\mathrm{in})}(\rho)-}{} \times
\left(
1+\frac{c_1(u_v)}{\sqrt{n}}+\frac{c_2(u_v)}{n}+\frac{c_3(u_v)}{n^{3/2}}+\mathcal{O}\left(\frac{c_4(u_v)}{n^2}\right)
\right) {\rm d}v
+{\rm e}^{-nV_{\tau}(r_{\tau})}\cdot\mathcal{O}\bigl({\rm e}^{-cM^2}\bigr),
\nonumber
\end{align}
where $u_v=\sqrt{nd_2}(\rho-r_{\tau})-v=\chi(\tau)-v$ with $\chi(\tau):=\sqrt{nd_2}(\rho-r_{\tau})$.
Since $\sqrt{d_2n}(\rho-r_{\tau}+\delta_n')>\sqrt{nd_2}(\rho-r_{\tau})$ for $g_{1,-}\leq j\leq \lfloor n\tau_{\rho}\rfloor-1$, one can safely extend the integral region to $[0,+\infty)$ with an error $\mathcal{O}({\rm e}^{-cM^2})$ for some $c>0$.
Note that for $a> -1$ and $k\in\mathbb{Z}_{\geq 0}$, by \eqref{def of parabolic cylinder function}, we have
\begin{equation*}
\int_{0}^{+\infty}
v^{a+k}{\rm e}^{-\frac{1}{2}(v-\chi(\tau))^2}\,{\rm d}v
=
\vartheta_{k,a}
{\rm e}^{-\frac{\chi(\tau)^2}{4}}D_{-a-k-1}(-\chi(\tau)).
\end{equation*}
Here, $\vartheta_{k,a}:=\Gamma(a+k+1)$.
Therefore, we have
\begin{align}
&
\int_{0}^{\sqrt{d_2n}(\rho-r_{\tau}+\delta_n')}
v^{a}
{\rm e}^{-\frac{1}{2}(v-\sqrt{nd_2}(\rho-r_{\tau}))^2}
\left(
1+\frac{c_1(u_v)}{\sqrt{n}}+\frac{c_2(u_v)}{n}+\frac{c_3(u_v)}{n^{3/2}}+\mathcal{O}\left(\frac{c_4(u_v)}{n^2}\right)
\right){\rm d}v\nonumber
\\
&\qquad=
\mathsf{c}_0+\frac{\mathsf{c}_1}{\sqrt{n}}+\frac{\mathsf{c}_2}{n}+\frac{\mathsf{c}_3}{n^{3/2}} +\mathcal{O}\left(\frac{\mathsf{c}_4}{n^2}\right),
\label{def of hin rho 2}
\end{align}
where $c_1$, $c_2$, $c_3$ are given by \eqref{def of c_1}, \eqref{def of c_2}, and \eqref{def of c_3}, respectively, and
\begin{gather}
\mathsf{c}_0 :=
\vartheta_{0,a}
{\rm e}^{-\frac{\chi(\tau)^2}{4}}D_{-a-1}(-\chi(\tau))
=
\vartheta_{0,a}
{\rm e}^{-\frac{\phi(\tau)^2}{4}}D_{-a-1}(\phi(\tau)),\label{def of sfc0}
\\
\mathsf{c}_1
 :=
\int_{0}^{+\infty}
v^{a}
{\rm e}^{-\frac{1}{2}(v-\sqrt{nd_2}(\rho-r_{\tau}))^2}c_1(u_v)\,{\rm d}v
\nonumber\\
 \hphantom{\mathsf{c}_1 \,}{} =
\frac{d_3}{6d_2^{3/2}}
\sum_{\ell=0}^{3}\binom{3}{\ell}\phi(\tau)^{3-\ell}\vartheta_{\ell,a}{\rm e}^{-\frac{\phi(\tau)^2}{4}}D_{-a-\ell-1}(\phi(\tau))
\nonumber\\
\hphantom{\mathsf{c}_0 :=}{} -
\frac{1}{d_2^{1/2}}\left(\mathsf{k}'(r_{\tau})+\frac{1}{r_{\tau}}\right)
\sum_{\ell=0}^{1}\binom{1}{\ell}\phi(\tau)^{1-\ell}\vartheta_{\ell,a}{\rm e}^{-\frac{\phi(\tau)^2}{4}}D_{-a-\ell-1}(\phi(\tau)),\label{def of sfc1}
\\
\mathsf{c}_2 :=\int_{0}^{+\infty}
v^{a}
{\rm e}^{-\frac{1}{2}(v-\sqrt{nd_2}(\rho-r_{\tau}))^2}c_2(u_v)\,{\rm d}v
\nonumber\\
\hphantom{\mathsf{c}_2\,}{} =
\frac{d_3^2}{72d_2^3}\sum_{\ell=0}^{6}\binom{6}{\ell}\phi(\tau)^{6-\ell}\vartheta_{\ell,a}{\rm e}^{-\frac{\chi(\tau)^2}{4}}D_{-a-\ell-1}(\phi(\tau))
\nonumber\\
\hphantom{\mathsf{c}_2 :=}{}
-\frac{1}{24d_2^2}
\left(d_4+4d_3\mathsf{k}'(r_{\tau})+\frac{4d_3}{r_{\tau}}\right)
\sum_{\ell=0}^{4}\binom{4}{\ell}\phi(\tau)^{4-\ell}
\vartheta_{\ell,a}{\rm e}^{-\frac{\phi(\tau)^2}{4}}D_{-a-\ell-1}(\phi(\tau))
\nonumber\\
\hphantom{\mathsf{c}_2 :=}{}
+\frac{1}{d_2}\left(\frac{\mathsf{k}''(r_{\tau})}{2}+\frac{\mathsf{k}'(r_{\tau})^2}{2} +\frac{\mathsf{k}'(r_{\tau})}{r_{\tau}}\right)
\sum_{\ell=0}^{2}\binom{2}{\ell}\phi(\tau)^{2-\ell}
\vartheta_{\ell,a}{\rm e}^{-\frac{\phi(\tau)^2}{4}}D_{-a-\ell-1}(\phi(\tau)),\label{def of sfc2}
\\
\mathsf{c}_3 :=\int_{0}^{+\infty}
v^{a}
{\rm e}^{-\frac{1}{2}(v-\sqrt{nd_2}(\rho-r_{\tau}))^2}c_3(u_v)\,{\rm d}v
\nonumber\\
\hphantom{\mathsf{c}_3\,}{} =
\frac{d_3^3}{1296d_2^{9/2}}
\sum_{\ell=0}^{9}\binom{9}{\ell}\phi(\tau)^{9-\ell}
\vartheta_{\ell,a}{\rm e}^{-\frac{\phi(\tau)^2}{4}}D_{-a-\ell-1}(\phi(\tau))
\nonumber\\
\hphantom{\mathsf{c}_3 :=}{}
-\left(d_3d_4+2d_3^2\mathsf{k}'(r_{\tau})+\frac{2d_3^2}{r_{\tau}} \right)\frac{1}{144d_2^{7/2}}
\sum_{\ell=0}^{7}\binom{7}{\ell}\phi(\tau)^{7-\ell}
\vartheta_{\ell,a}{\rm e}^{-\frac{\phi(\tau)^2}{4}}D_{-a-\ell-1}(\phi(\tau))
\nonumber\\
\hphantom{\mathsf{c}_3 :=}{}
+
\left(
d_5+\frac{5}{r_{\tau}}d_4+5d_4\mathsf{k}'(r_{\tau})+\frac{20d_3}{r_{\tau}}\mathsf{k}'(r_{\tau}) +10d_3\mathsf{k}'(r_{\tau})^2+10d_3\mathsf{k}''(r_{\tau})
\right)
\nonumber\\
\hphantom{\mathsf{c}_3 :=}{} \times
\sum_{\ell=0}^{5}\binom{5}{\ell}
\frac{\phi(\tau)^{5-\ell}\vartheta_{\ell,a}{\rm e}^{-\frac{\phi(\tau)^2}{4}}}{120d_2^{5/2}}D_{-a-\ell-1}(\phi(\tau))
\nonumber\\
\hphantom{\mathsf{c}_3 :=}{}
-
\left(\frac{3\mathsf{k}'(r_{\tau})^2}{r_{\tau}}
+\frac{3\mathsf{k}''(r_{\tau})}{r_{\tau}}+\mathsf{k}'(r_{\tau})^3+3\mathsf{k}'(r_{\tau})\mathsf{k}''(r_{\tau}) +\mathsf{k}^{(3)}(r_{\tau})\right)
\nonumber\\
\hphantom{\mathsf{c}_3 :=}{} \times
\sum_{\ell=0}^{3}\binom{3}{\ell}
\frac{\phi(\tau)^{3-\ell}
\vartheta_{\ell,a}{\rm e}^{-\frac{\phi(\tau)^2}{4}}}{6d_2^{3/2}}D_{-a-\ell-1}(\phi(\tau)).\label{def of sfc3}
\end{gather}
We next consider $h_{n,j}^{(\mathrm{out})}(\rho)$.
In this integral region and for $g_{1,-}\leq j\leq \lfloor n\tau_{\rho} \rfloor -1$, there is no critical point, but we still proceed with the expansion in terms of $r_{\tau}$. Indeed,
\begin{align*}
h_{n,j}^{(\mathrm{out})}(\rho)
&=\int_{\rho}^{+\infty} 2v{\rm e}^{\mathsf{k}(v)}{\rm e}^{-nV_{\tau}(v)}|v-\rho|^{a}\,{\rm d}v
\\
&=
\int_{0}^{+\infty} 2v{\rm e}^{\mathsf{k}(v)}{\rm e}^{-nV_{\tau}(v)}|v-\rho|^{a}\,{\rm d}v
-
\int_0^{\rho}2v{\rm e}^{\mathsf{k}(v)}{\rm e}^{-nV_{\tau}(v)}|v-\rho|^{a}\,{\rm d}v.
\end{align*}
We already know the asymptotics of the second terms in the last line.
For the first term, we split the integral $[0,+\infty)$ into $[0,+\infty)=\{r\in[0,+\infty) \mid |r-r_{\tau}|<\delta_n'\}\cup\{r\in[0,+\infty) \mid |r-r_{\tau}|\geq\delta_n'\}$.
In the latter region, the corresponding integral can be neglected with an error ${\rm e}^{-nV_{\tau}(r_{\tau})}\cdot \mathcal{O}\bigl({\rm e}^{-cM^2}\bigr)$ for some $c>0$.
The Gaussian integral gives
\begin{align*}
&
\frac{2r_{\tau}{\rm e}^{\mathsf{k}(r_{\tau})}{\rm e}^{-nV_{\tau}(r_{\tau})}}{(nd_2)^{1/2}}
\int_{-\sqrt{d_2n}\delta_n'}^{\sqrt{d_2n}\delta_n'}
{\rm e}^{-\frac{1}{2}u^2}
\left|r_{\tau}-\rho+\frac{u}{\sqrt{nd_2}}\right|^{a}
\\
&\qquad \quad{}\times
\left(
1+\frac{c_1(u)}{\sqrt{n}}+\frac{c_2(u)}{n}+\frac{c_3(u)}{n^{3/2}}+\mathcal{O}\left(\frac{c_4(u)}{n^2}\right)
\right){\rm d} u
\\
&\qquad{} =
\frac{2r_{\tau}{\rm e}^{\mathsf{k}(r_{\tau})}{\rm e}^{-nV_{\tau}(r_{\tau})}}{(nd_2)^{(a+1)/2}}
\int_{-\sqrt{d_2n}\delta_n'+\sqrt{d_2n}(r_{\tau}-\rho)}^{\sqrt{nd_2}\delta_n'+\sqrt{nd_2}(r_{\tau}-\rho)}
{\rm e}^{-\frac{1}{2}(v-\sqrt{nd_2}(r_{\tau}-\rho))^2}
|v|^{a}
\\
&\qquad\quad{}\times
\left(
1+\frac{c_1(u_v)}{\sqrt{n}}+\frac{c_2(u_v)}{n}+\frac{c_3(u_v)}{n^{3/2}}+\mathcal{O}\left(\frac{c_4(u_v)}{n^2}\right)
\right){\rm d} v,
\end{align*}
where $u_v=v-\sqrt{nd_2}(r_{\tau}-\rho)$.
One can safely extend the integral region to $(-\infty,+\infty)$ with an exponential error $\mathcal{O}\bigl({\rm e}^{-cM^2}\bigr)$ for some $c>0$.
Each integral is given by the following: for $\phi(\tau):=\sqrt{nd_2}(r_{\tau}-\rho)$,
\begin{align*}
\int_{-\infty}^{+\infty}
{\rm e}^{-\frac{1}{2}(v-\sqrt{nd_2}(r_{\tau}-\rho))^2}
|v|^{a}\,{\rm d}v
&=
\int_{0}^{+\infty}
{\rm e}^{-\frac{1}{2}(v-\phi(\tau))^2}
v^{a}\,{\rm d}v
+\int_{0}^{+\infty}
{\rm e}^{-\frac{1}{2}(v+\phi(\tau))^2}
v^{a}\,{\rm d}v
\\
&
=\vartheta_{0,a}{\rm e}^{-\frac{1}{4}\phi(\tau)^2}\bigl(
D_{-a-1}(-\phi(\tau))+D_{-a-1}(\phi(\tau))
\bigr).
\end{align*}
Since by \eqref{def of parabolic cylinder function},
\begin{gather*}
\int_{0}^{+\infty}t^{a}{\rm e}^{-\frac{1}{2}(t+z)^2}\,{\rm d}t
 =\Gamma(a+1){\rm e}^{-\frac{1}{4}z^2}D_{-a-1}(z),
\\
\int_{-\infty}^{+\infty}{\rm e}^{-\frac{1}{2}(u-\phi(\tau))^2}|v|^{a}v^{\ell}\,{\rm d}v\\
 \qquad{} =
\Gamma(a+\ell+1){\rm e}^{-\frac{1}{4}\phi(\tau)^2}
\bigl(
(-1)^{\ell}D_{-a-\ell-1}(\phi(\tau))
+
D_{-a-\ell-1}(-\phi(\tau))
\bigr),
\end{gather*}
we have
\begin{gather}
\label{def of sfcd0}
\mathsf{c}_0'=
\int_{-\infty}^{+\infty}{\rm e}^{-\frac{1}{2}(u-\phi(\tau))^2}|v|^{a}\,{\rm d}v
=
\vartheta_{0,a}{\rm e}^{-\frac{1}{4}\phi(\tau)^2}
\bigl(
D_{-a-1}(\phi(\tau))
+
D_{-a-1}(-\phi(\tau))
\bigr),
\\
\mathsf{c}_1' =\int_{-\infty}^{+\infty}
{\rm e}^{-\frac{1}{2}(v-\sqrt{nd_2}(r_{\tau}-\rho))^2}
|v|^{a}
c_1(u_v)\,{\rm d}v
\nonumber\\
\hphantom{\mathsf{c}_1'}{} =
-\frac{d_3}{6d_2^{3/2}}\sum_{\ell=0}^{3}\binom{3}{\ell}(-\phi(\tau))^{3-\ell}
\vartheta_{\ell,a}{\rm e}^{-\frac{1}{4}\phi(\tau)^2}
\bigl(
(-1)^{\ell}D_{-a-\ell-1}(\phi(\tau))
+
D_{-a-\ell-1}(-\phi(\tau))
\bigr)
\nonumber
\\ \hphantom{\mathsf{c}_1'=}{}
+\frac{1}{d_2^{1/2}}\Bigl(\mathsf{k}'(r_{\tau})+\frac{1}{r_{\tau}}\Bigr)
\nonumber
\\
\hphantom{\mathsf{c}_1'=}{}\times
\sum_{\ell=0}^{1}\binom{1}{\ell}(-\phi(\tau))^{1-\ell}
\vartheta_{\ell,a}{\rm e}^{-\frac{1}{4}\phi(\tau)^2}
\bigl(
(-1)^{\ell}D_{-a-\ell-1}(\phi(\tau))
+
D_{-a-\ell-1}(-\phi(\tau))
\bigr),\label{def of sfcd1}
\\
\mathsf{c}_2' = \int_{-\infty}^{+\infty}
{\rm e}^{-\frac{1}{2}(v-\sqrt{nd_2}(r_{\tau}-\rho))^2}
|v|^{a}
c_2(u_v)\,{\rm d}v
\nonumber\\
\hphantom{\mathsf{c}_2'}{} =
\frac{d_3^2}{72d_2^3}
\sum_{\ell=0}^{6}\binom{6}{\ell}(-\phi(\tau))^{6-\ell}
\vartheta_{\ell,a}{\rm e}^{-\frac{1}{4}\phi(\tau)^2}
\bigl(
(-1)^{\ell}D_{-a-\ell-1}(\phi(\tau))
+
D_{-a-\ell-1}(-\phi(\tau))
\bigr)
\nonumber
\\
\hphantom{\mathsf{c}_2' =}{}
-
\frac{1}{24d_2^2}\left(d_4+4d_3\mathsf{k}'(r_{\tau})+\frac{4d_3}{r_{\tau}}\right)
\nonumber
\\
\hphantom{\mathsf{c}_2' =}{}\times
\sum_{\ell=0}^{4}\binom{4}{\ell}(-\phi(\tau))^{4-\ell}
\vartheta_{\ell,a}{\rm e}^{-\frac{1}{4}\phi(\tau)^2}
\bigl(
(-1)^{\ell}D_{-a-\ell-1}(\phi(\tau))
+
D_{-a-\ell-1}(-\phi(\tau))
\bigr)
\nonumber
\\
\hphantom{\mathsf{c}_2' =}{}
+
\frac{1}{d_2}\left(\frac{\mathsf{k}''(r_{\tau})}{2}+\frac{\mathsf{k}'(r_{\tau})^2}{2}+\frac{\mathsf{k}'(r_{\tau})}{r_{\tau}}\right)
\nonumber
\\
\hphantom{\mathsf{c}_2' =}{}\times
\sum_{\ell=0}^{2}\binom{2}{\ell}(-\phi(\tau))^{2-\ell}
\vartheta_{\ell,a}{\rm e}^{-\frac{1}{4}\phi(\tau)^2}
\bigl(
(-1)^{\ell}D_{-a-\ell-1}(\phi(\tau))
+
D_{-a-\ell-1}(-\phi(\tau))
\bigr).
\label{def of sfcd2}
\end{gather}
Note that $\mathsf{c}'_3$ is given as well as \eqref{def of sfc3}, but we omit the explicit expression.
Therefore, we obtain
\begin{align}
h_{n,j}^{(\mathrm{out})}(\rho)
&=
\int_{0}^{+\infty} 2v{\rm e}^{\mathsf{k}(v)}{\rm e}^{-nV_{\tau}(v)}|v-\rho|^{a}\,{\rm d}v
-
\int_0^{\rho}2v{\rm e}^{\mathsf{k}(v)}{\rm e}^{-nV_{\tau}(v)}|v-\rho|^{a}\,{\rm d}v
\nonumber\\
&=
\frac{2r_{\tau}{\rm e}^{\mathsf{k}(r_{\tau})}{\rm e}^{-nV_{\tau}(r_{\tau})}}{(nd_2)^{(a+1)/2}}
\left(
\mathsf{c}_0'-\mathsf{c}_0+\frac{\mathsf{c}_1'-\mathsf{c}_1}{\sqrt{n}}+\frac{\mathsf{c}_2' -\mathsf{c}_2}{n}+\mathcal{O}\left(\frac{\mathsf{c}_3'-\mathsf{c}_3}{n^{3/2}}\right)
\right),\label{def of hrho out 2}
\end{align}
where by \eqref{def of sfc0}, \eqref{def of sfc1}, \eqref{def of sfc2}, \eqref{def of sfc3}, \eqref{def of sfcd0}, \eqref{def of sfcd1}, and \eqref{def of sfcd2},
\begin{gather*}
\mathsf{c}_0'-\mathsf{c}_0
 =
\vartheta_{0,a}{\rm e}^{-\frac{1}{4}\phi(\tau)^2}
D_{-a-1}(-\phi(\tau)),
\\
\mathsf{c}_1'-\mathsf{c}_1
 =
-\frac{d_3}{6d_2^{3/2}}\sum_{\ell=0}^{3}\binom{3}{\ell}(-\phi(\tau))^{3-\ell}
\vartheta_{\ell,a}{\rm e}^{-\frac{1}{4}\phi(\tau)^2}
D_{-a-\ell-1}(-\phi(\tau))
\\
\hphantom{\mathsf{c}_1'-\mathsf{c}_1=}{} +\frac{1}{d_2^{1/2}}\left(\mathsf{k}'(r_{\tau})+\frac{1}{r_{\tau}}\right)
\sum_{\ell=0}^{1}\binom{1}{\ell}(-\phi(\tau))^{1-\ell}
\vartheta_{\ell,a}{\rm e}^{-\frac{1}{4}\phi(\tau)^2}
D_{-a-\ell-1}(-\phi(\tau)),
\\
\mathsf{c}_2'-\mathsf{c}_2 =
\frac{d_3^2}{72d_2^3}
\sum_{\ell=0}^{6}\binom{6}{\ell}(-\phi(\tau))^{6-\ell}
\vartheta_{\ell,a}{\rm e}^{-\frac{1}{4}\phi(\tau)^2}
D_{-a-\ell-1}(-\phi(\tau))
\\
\hphantom{\mathsf{c}_2'-\mathsf{c}_2 =}{} -
\frac{1}{24d_2^2}\left(d_4+4d_3\mathsf{k}'(r_{\tau})+\frac{4d_3}{r_{\tau}}\right)
\sum_{\ell=0}^{4}\binom{4}{\ell}(-\phi(\tau))^{4-\ell}
\vartheta_{\ell,a}{\rm e}^{-\frac{1}{4}\phi(\tau)^2}
D_{-a-\ell-1}(-\phi(\tau))
\\
\hphantom{\mathsf{c}_2'-\mathsf{c}_2 =}{} +
\frac{1}{d_2}\left(\frac{\mathsf{k}''(r_{\tau})}{2}+\frac{\mathsf{k}'(r_{\tau})^2}{2}+\frac{\mathsf{k}'(r_{\tau})}{r_{\tau}}\right)\\
\hphantom{\mathsf{c}_2'-\mathsf{c}_2 =}{}\times
\sum_{\ell=0}^{2}\binom{2}{\ell}(-\phi(\tau))^{2-\ell}
\vartheta_{\ell,a}{\rm e}^{-\frac{1}{4}\phi(\tau)^2}
D_{-a-\ell-1}(-\phi(\tau)).
\end{gather*}
By \eqref{def of hin rho 1}, \eqref{def of hin rho 2}, and \eqref{def of hrho out 2}, let us choose $\delta>0$ sufficiently small so that ${\rm e}^{u}\frac{h_{n,j}^{(\mathrm{in})}(\rho)}{h_{n,j}}+\frac{h_{n,j}^{(\mathrm{out})}(\rho)}{h_{n,j}}$ remains bounded away from the interval $(-\infty,0]$ as $n\to+\infty$ uniformly for $u\in\{z\in\C \mid |z-x|\leq \delta\}$ and $a$ in compact subsets of $(-1,+\infty)$.
Then, we have
\begin{align}
& \sum_{j=g_{1,-}}^{\lfloor n\tau_{\rho}\rfloor -1}
\log\left({\rm e}^{u}\frac{h_{n,j}^{(\mathrm{in})}(\rho)}{h_{n,j}}+\frac{h_{n,j}^{(\mathrm{out})}(\rho)}{h_{n,j}}\right)
\nonumber\\
&\qquad =
\sum_{j=g_{1,-}}^{\lfloor n\tau_{\rho}\rfloor -1}\log\left(\frac{1}{(4n)^{a/2}}\right)
+
\frac{a}{2}\sum_{j=g_{1,-}}^{\lfloor n\tau_{\rho}\rfloor -1}\log\left(\frac{1}{\Delta Q(\rho)}\right)
+
\sum_{j=g_{1,-}}^{\lfloor n\tau_{\rho}\rfloor -1}\log [\mathcal{H}_{a,u}(\xi) ]
\nonumber\\
&\qquad\quad{}
+\frac{1}{\sqrt{n}} \sum_{j=g_{1,-}}^{\lfloor n\tau_{\rho}\rfloor -1}\frac{\widetilde{g}_{a,u}(\xi)}{12\rho\sqrt{\Delta Q(\rho)}g_{a,u}(\xi)}
+\mathcal{O}\left(\frac{(\log n)^{7/8}}{n^{1/8}}\right),\label{def of S2 first part}
\end{align}
uniformly for $u\in\{z\in\C\mid |z-x|\leq \delta\}$ and $a$ in compact subsets of $(-1,+\infty)$.
Here, $g_{a,u}(x),\widetilde{g}_{a,u}(x)$ are given by \eqref{def of cal H au} and \eqref{def of tilde gau}, respectively.

We consider the case $\lfloor n\tau_{\rho}\rfloor\leq j\leq g_{1,+}$.
Similar manner to \eqref{def of hrho out 2}, there exists $c>0$ such that we have
\begin{align*}
& h_{n,j}^{(\mathrm{out})}(\rho) =
\frac{2r_{\tau}{\rm e}^{\mathsf{k}(r_{\tau})}{\rm e}^{-nV_{\tau}(r_{\tau})}}{(nd_2)^{(a+1)/2}}
\int_{0}^{\sqrt{nd_2}\delta_n'+\sqrt{nd_2}(r_{\tau}-\rho)}
y^a
{\rm e}^{-\frac{1}{2}(y-\sqrt{nd_2}(r_{\tau}-\rho))^2}
\\
&\hphantom{h_{n,j}^{(\mathrm{out})}(\rho) =}{}
\times
\left(
1+\frac{c_1(u_y)}{\sqrt{n}}+\frac{c_2(u_y)}{n}+\frac{c_3(u)}{n^{3/2}}+\mathcal{O}\left(\frac{c_4(u_y)}{n^2}\right)
\right) {\rm d}y
+\mathcal{O}\bigl({\rm e}^{-cM^2}\bigr),
\end{align*}
where $u_y:=y-\sqrt{nd_2}(r_{\tau}-\rho)=y-\phi(\tau)$ with $\phi(\tau):=\sqrt{nd_2}(r_{\tau}-\rho)$.
Since $\sqrt{nd_2}\delta_n'+\sqrt{nd_2}(r_{\tau}-\rho)>\sqrt{nd_2}(r_{\tau}-\rho)$ for $\lfloor n\tau_{\rho}\rfloor\leq j \leq g_{1,+}$, we can extend the integral region to $[0,+\infty)$ with an exponential error $\mathcal{O}\bigl({\rm e}^{-cM^2}\bigr)$.
The Gaussian integral gives
\begin{equation}
\label{def of hrho out second part}
h_{n,j}^{(\mathrm{out})}(\rho)
=
\frac{2r_{\tau}{\rm e}^{\mathsf{k}(r_{\tau})}{\rm e}^{-nV_{\tau}(r_{\tau})}}{(nd_2)^{(a+1)/2}}
\left(
\widetilde{\mathsf{c}}_0
+\frac{\widetilde{\mathsf{c}}_1}{\sqrt{n}}
+\frac{\widetilde{\mathsf{c}}_2}{n}
+\frac{\widetilde{\mathsf{c}}_3}{n^{3/2}}
+\mathcal{O}\left(\frac{\widetilde{\mathsf{c}}_4}{n^2}\right)
\right)
+\mathcal{O}\bigl({\rm e}^{-cM^2}\bigr),
\end{equation}
where
\begin{gather*}
\widetilde{\mathsf{c}}_0
 =
\int_{0}^{+\infty}
y^a{\rm e}^{-\frac{1}{2}(y-\phi(\tau))^2}\,{\rm d}y
=
\vartheta_{0,a}
{\rm e}^{-\frac{\phi(\tau)^2}{4}}D_{-a-1}(-\phi(\tau)),
\\
\widetilde{\mathsf{c}}_1
 =\int_{0}^{+\infty}
y^a{\rm e}^{-\frac{1}{2}(y-\phi(\tau))^2}c_1(u_y)\,{\rm d}y
\\
 \hphantom{\widetilde{\mathsf{c}}_1}{} =
-\frac{d_3}{6d_2^{3/2}}
\sum_{\ell=0}^{3}\binom{3}{\ell}(-\phi(\tau))^{3-\ell}\vartheta_{\ell,a}{\rm e}^{-\frac{\phi(\tau)^2}{4}}D_{-a-\ell-1}(-\phi(\tau))
\\
 \hphantom{\widetilde{\mathsf{c}}_1=}{}
+
\frac{1}{d_2^{1/2}}\Bigl(\mathsf{k}'(r_{\tau})+\frac{1}{r_{\tau}}\Bigr)
\sum_{\ell=0}^{1}\binom{1}{\ell}(-\phi(\tau))^{1-\ell}\vartheta_{\ell,a}{\rm e}^{-\frac{\phi(\tau)^2}{4}}D_{-a-\ell-1}(-\phi(\tau)),
\\
\widetilde{\mathsf{c}}_2
 =
\int_{0}^{+\infty}
y^a{\rm e}^{-\frac{1}{2}(y-\phi(\tau))^2}c_2(u_y)\,{\rm d}y
\\
\hphantom{\widetilde{\mathsf{c}}_2}{} =
\frac{d_3^2}{72d_2^3}
\sum_{\ell=0}^{6}\binom{6}{\ell}(-\phi(\tau))^{6-\ell}\vartheta_{\ell,a}{\rm e}^{-\frac{\phi(\tau)^2}{4}}D_{-a-\ell-1}(-\phi(\tau))
\\
\hphantom{\widetilde{\mathsf{c}}_2=}{}
-\frac{1}{24d_2^2}\left(d_4+4d_3\mathsf{k}'(r_{\tau})+\frac{4d_3}{r_{\tau}}\right)
\sum_{\ell=0}^{4}\binom{4}{\ell}(-\phi(\tau))^{4-\ell}\vartheta_{\ell,a}{\rm e}^{-\frac{\phi(\tau)^2}{4}}D_{-a-\ell-1}(-\phi(\tau))
\\
\hphantom{\widetilde{\mathsf{c}}_2=}{}
+\frac{1}{d_2}\left(\frac{\mathsf{k}''(r_{\tau})}{2}+\frac{\mathsf{k}'(r_{\tau})^2}{2} +\frac{\mathsf{k}'(r_{\tau})}{r_{\tau}}\right)
\sum_{\ell=0}^{2}\binom{2}{\ell}(-\phi(\tau))^{2-\ell}\vartheta_{\ell,a}{\rm e}^{-\frac{\phi(\tau)^2}{4}}D_{-a-\ell-1}(-\phi(\tau)).
\end{gather*}
On the other hand,
\begin{align}
h_{n,j}^{(\mathrm{in})}(\rho)
:={}&
\int_{0}^{+\infty} 2v{\rm e}^{\mathsf{k}(v)}{\rm e}^{-nV_{\tau}(v)}|v-\rho|^{a}\,{\rm d}v
-
\int_{\rho}^{+\infty} 2v{\rm e}^{\mathsf{k}(v)}{\rm e}^{-nV_{\tau}(v)}(v-\rho)
^{a}\,{\rm d}v
\nonumber\\
 ={}&
\frac{2r_{\tau}{\rm e}^{\mathsf{k}(r_{\tau})}{\rm e}^{-nV_{\tau}(r_{\tau})}}{(nd_2)^{(a+1)/2}}
\left(
\mathsf{c}_0'-\widetilde{\mathsf{c}}_0+\frac{\mathsf{c}_1'-\widetilde{\mathsf{c}}_1}{\sqrt{n}} +\frac{\mathsf{c}_2'-\widetilde{\mathsf{c}}_2}{n}+\mathcal{O}\Bigl(\frac{\mathsf{c}_3' -\widetilde{\mathsf{c}}_3}{n^{3/2}}\right)
\Bigr),\label{def of hrho in second part}
\end{align}
where
\begin{gather*}
\mathsf{c}_0'-\widetilde{\mathsf{c}}_0
 =
\vartheta_{0,a}{\rm e}^{-\frac{1}{4}\phi(\tau)^2}
D_{-a-1}(\phi(\tau)),
\\
\mathsf{c}_1'-\widetilde{\mathsf{c}}_1 =
-\frac{d_3}{6d_2^{3/2}}\sum_{\ell=0}^{3}\binom{3}{\ell}(-\phi(\tau))^{3-\ell}
\vartheta_{\ell,a}{\rm e}^{-\frac{1}{4}\phi(\tau)^2}
(-1)^{\ell}D_{-a-\ell-1}(\phi(\tau))
\\
\hphantom{\mathsf{c}_1'-\widetilde{\mathsf{c}}_1 =}{}
+\frac{1}{d_2^{1/2}}\left(\mathsf{k}'(r_{\tau})+\frac{1}{r_{\tau}}\right)
\sum_{\ell=0}^{1}\binom{1}{\ell}(-\phi(\tau))^{1-\ell}
\vartheta_{\ell,a}{\rm e}^{-\frac{1}{4}\phi(\tau)^2}
(-1)^{\ell}D_{-a-\ell-1}(\phi(\tau)),
\\
\mathsf{c}_2'-\widetilde{\mathsf{c}}_2 =
\frac{d_3^2}{72d_2^3}
\sum_{\ell=0}^{6}\binom{6}{\ell}(-\phi(\tau))^{6-\ell}
\vartheta_{\ell,a}{\rm e}^{-\frac{1}{4}\phi(\tau)^2}
(-1)^{\ell}D_{-a-\ell-1}(\phi(\tau))
\\
\hphantom{\mathsf{c}_2'-\widetilde{\mathsf{c}}_2 =}{}
-
\frac{1}{24d_2^2}\left(d_4+4d_3\mathsf{k}'(r_{\tau})+\frac{4d_3}{r_{\tau}}\right)\\
\hphantom{\mathsf{c}_2'-\widetilde{\mathsf{c}}_2 =}{}\times
\sum_{\ell=0}^{4}\binom{4}{\ell}(-\phi(\tau))^{4-\ell}
\vartheta_{\ell,a}{\rm e}^{-\frac{1}{4}\phi(\tau)^2}
(-1)^{\ell}D_{-a-\ell-1}(\phi(\tau))
\\
\hphantom{\mathsf{c}_2'-\widetilde{\mathsf{c}}_2 =}{}
+
\frac{1}{d_2}\left(\frac{\mathsf{k}''(r_{\tau})}{2}+\frac{\mathsf{k}'(r_{\tau})^2}{2}+\frac{\mathsf{k}'(r_{\tau})}{r_{\tau}}\right)\\
\hphantom{\mathsf{c}_2'-\widetilde{\mathsf{c}}_2 =}{}\times
\sum_{\ell=0}^{2}\binom{2}{\ell}(-\phi(\tau))^{2-\ell}
\vartheta_{\ell,a}{\rm e}^{-\frac{1}{4}\phi(\tau)^2}
(-1)^{\ell}D_{-a-\ell-1}(\phi(\tau)).
\end{gather*}
Rewriting the above expansion in terms of \eqref{def of xi}, by \eqref{def of hrho out second part} and \eqref{def of hrho in second part} we obtain the same expansion with \eqref{def of S2 first part}.
\end{proof}

To apply Lemma~\ref{lemma:Euler--Maclaurin} to each sum in Lemma~\ref{lemma:S2in part1}, we will use
\begin{equation}
\label{def of d relationship between g0 g1}
\partial_x\log [\mathcal{H}_{a,u}(x) ]
=
-\frac{{\rm e}^{u}D_{-a}(x)-D_{-a}(-x)}{{\rm e}^uD_{-a-1}(x)+D_{-a-1}(-x)},
\end{equation}
where we have used \cite[Sections~12.8 and~12.9]{NIST}.
Now, we establish the asymptotic expansion of $S_2$.
\begin{Lemma}\label{lemma:Analysis of S2}
For any fixed $x\in\R$, there exists $\delta>0$ such that as $n\to+\infty$, we have
\begin{align*}
S_2 =
C_2^{(2)}\sqrt{n}
+C_3^{(2)}
+C_n^{(2)}
+\mathcal{O}\left(\frac{(\log n)^{1/4}}{n^{1/4}}\right),
\end{align*}
uniformly for $u\in\{z\in\C \mid |z-x|\leq \delta\}$ and $a$ in compact subsets of $(-1,+\infty)$, where
\begin{gather*}
C_2^{(2)} :=
\rho\sqrt{\Delta Q(\rho)}
\int_{-\infty}^{\infty}
\bigl(\log [\mathcal{H}_{a,u}(x) ]
-a\log|x|-u\mathbf{1}_{(-\infty,0)}(x)
\bigr)
\,{\rm d}x,
\\
C_3^{(2)} :=
-
\left(\alpha+\frac{1}{2}\right)u
+
\frac{1}{6}\left(2+\frac{\rho\partial_r\Delta Q(\rho)}{\Delta Q(\rho)}\right)u
-
\frac{a}{12}\left(1-\frac{\rho\partial_r\Delta Q(\rho)}{\Delta Q(\rho)}\right)u
\\
\hphantom{C_3^{(2)} :=}{}
+\frac{1}{6}\left(2+\frac{\rho\partial_r\Delta Q(\rho)}{\Delta Q(\rho)}\right)\\
\hphantom{C_3^{(2)} :=}{}\times
\int_{-\infty}^{+\infty}
\left[
x\bigl(\log [\mathcal{H}_{a,u}(x) ]
- u\mathbf{1}_{(-\infty,0)}
\bigr)
-ax\log|x|
-\frac{a(a-1)x}{2(x^2+1)}
\right]\,{\rm d}x,
\nonumber
\\
C_n^{(2)} :=
-\frac{a}{2}\bigl(\rho\sqrt{\Delta Q(\rho)}\sqrt{n}\Delta_{n,+}
+\rho\sqrt{\Delta Q(\rho)}\sqrt{n}\Delta_{n,-}
+1\bigr)\log(4n)
\\
\hphantom{C_n^{(2)} :=}{}
+
\sqrt{n}\rho\sqrt{\Delta Q(\rho)}
\bigl(
a \Delta_{n,+}(\log \Delta_{n,+}-1)
+a \Delta_{n,-}(\log \Delta_{n,-}-1)
+u\Delta_{n,-}
\bigr)
\\
\hphantom{C_n^{(2)} :=}{}
+\frac{a}{2}\log\Delta_{n,+}
+\frac{a}{2}\log\Delta_{n,-}.
\end{gather*}
\end{Lemma}

\begin{proof}
By \eqref{def of critical indices}, as $n\to+\infty$, we have
\[
g_{1,+}-g_{1,-}+1
=
\rho\sqrt{\Delta Q(\rho)}\sqrt{n}\Delta_{n,+}
+\rho\sqrt{\Delta Q(\rho)}\sqrt{n}\Delta_{n,-}
+1.
\]
Therefore, the first and second terms of \eqref{def of S2 asymptotics part 1} satisfy
\begin{align*}
&
\sum_{j=g_{1,-}}^{g_{1,+}}\log\left(\frac{1}{(4n)^{a/2}}\right)
+
\frac{a}{2}
\sum_{j=g_{1,-}}^{g_{1,+}}\log\left(\frac{1}{\Delta Q(\rho)}\right)
\\
&\qquad{}=
-\bigl( \rho\sqrt{\Delta Q(\rho)}\sqrt{n}\Delta_{n,+}
+\rho\sqrt{\Delta Q(\rho)}\sqrt{n}\Delta_{n,-}
+1\bigr)
a\log2
\\
&\quad\qquad{}
-\bigl( \rho\sqrt{\Delta Q(\rho)}\sqrt{n}\Delta_{n,+}
+\rho\sqrt{\Delta Q(\rho)}\sqrt{n}\Delta_{n,-}
+1\bigr)\frac{a}{2}\log n
\\
&\quad\qquad{}
-\bigl( \rho\sqrt{\Delta Q(\rho)}\sqrt{n}\Delta_{n,+}
+\rho\sqrt{\Delta Q(\rho)}\sqrt{n}\Delta_{n,-}
+1\bigr)\frac{a}{2}\log \Delta Q(\rho).
\end{align*}
By Lemma~\ref{lemma:Euler--Maclaurin}, as $n\to+\infty$, we have
\begin{align*}
&\sum_{j=g_{1,-}}^{g_{1,+}}
\log [\mathcal{H}_{a,u}(\xi) ]
=
\int_{g_{1,-}}^{g_{1,+}}\log\bigl[\mathcal{H}_{a,u}(\xi_t)\bigr]\,{\rm d}t
+\frac{1}{2}\log [\mathcal{H}_{a,u}(-\Delta_{n,-}) ]
\\
&\hphantom{\sum_{j=g_{1,-}}^{g_{1,+}}
\log [\mathcal{H}_{a,u}(\xi) ]=}{}
+\frac{1}{2}\log [\mathcal{H}_{a,u}(\Delta_{n,+}) ]
+\mathcal{O}\left(\frac{(\log n)^{1/4}}{n^{1/4}}\right).
\end{align*}
By change of variables \eqref{def of rtau ODE}, we get
\[
\int_{g_{1,-}}^{g_{1,+}}\log [\mathcal{H}_{a,u}(\xi_t) ]\,{\rm d}t
=
\sqrt{n}\,\rho\sqrt{\Delta Q(\rho)}
\int_{-\Delta_{n,-}}^{\Delta_{n,+}}
\log [\mathcal{H}_{a,u}(x) ]\,{\rm d}x.
\]
From \cite[Section~12.9]{NIST}, we have
\begin{equation}\label{def of asymptotics of mathcal H au x}
\log [\mathcal{H}_{a,u}(x) ]
=
a\log |x|+u\mathbf{1}_{(-\infty,0)}(x)+\frac{a(a-1)}{2x^2}-\frac{a(a-1)(2a-3)}{4x^4}+\mathcal{O}\bigl(x^{-6}\bigr),
\end{equation}
for $x\to\pm\infty$. By the above and regularizing the integral, we have
\begin{align*}
& \sqrt{n}\,\rho\sqrt{\Delta Q(\rho)}
\int_{-\Delta_{n,-}}^{\Delta_{n,+}}
\log [\mathcal{H}_{a,u}(x) ]\,{\rm d}x
\\
&\qquad{}=
\sqrt{n} \rho\sqrt{\Delta Q(\rho)}
\int_{-\infty}^{+\infty}
\bigl(\log [\mathcal{H}_{a,u}(x) ]
-a\log|x|-u\mathbf{1}_{(-\infty,0)}(x)
\bigr)
\,{\rm d}x
\\
&\qquad\quad{}
+\sqrt{n}\rho\sqrt{\Delta Q(\rho)}
\bigl(
a \Delta_{n,+}(\log \Delta_{n,+}-1)
+a \Delta_{n,-}(\log \Delta_{n,-}-1)
+u\Delta_{n,-}
\bigr)
\\
&\quad\qquad{}
-\sqrt{n}\rho\sqrt{\Delta Q(\rho)}\biggl(
\frac{a(a-1)}{2\Delta_{n,+}}
+\frac{a(a-1)}{2\Delta_{n,-}}
-\frac{a(a-1)(2a-3)}{12\Delta_{n,+}^3}
-\frac{a(a-1)(2a-3)}{12\Delta_{n,-}^3}
\\
&\quad\qquad{}
+\mathcal{O}\bigl(\Delta_{n,+}^{-5}+\Delta_{n,-}^{-5}\bigr)
\biggr).
\end{align*}
Next we observe that
\begin{align*}
&\frac{1}{\sqrt{n}} \sum_{j=g_{1,-}}^{g_{1,+}}
\frac{\widetilde{g}_{a,u}(\xi)}{12\rho\sqrt{\Delta Q(\rho)}g_{a,u}(\xi)}
=
\frac{1}{12}\int_{-\Delta_{n,-}}^{\Delta_{n,+}}
\frac{\widetilde{g}_{a,u}(x)}{g_{a,u}(x)}\,{\rm d}x
+\mathcal{O}\left(\frac{(\log n)^{1/4}}{n^{1/4}}\right),
\end{align*}
Note that by \eqref{def of tilde gau}, \eqref{def of cal H au}, and \eqref{def of d relationship between g0 g1}, we have
\begin{align*}
&
\frac{1}{12}\int_{-\Delta_{n,-}}^{\Delta_{n,+}}
\frac{\widetilde{g}_{a,u}(x)}{g_{a,u}(x)}\,{\rm d}x
\\
&\quad =
\frac{1}{6}\left(2+\frac{\rho \partial_r\Delta Q(\rho)}{\Delta Q(\rho)}\right)
\int_{-\Delta_{n,-}}^{\Delta_{n,+}}
x\log [\mathcal{H}_{a,u}(x)\ ]\,{\rm d}x
-\frac{a}{24}\left(1-\frac{\rho \partial_r\Delta Q(\rho)}{\Delta Q(\rho)}\right)\bigl(\Delta_{n,+}^2-\Delta_{n,-}^2\bigr)
\\
&
\quad\quad{}
-\frac{1}{12}\left(2+\frac{\rho \partial_r\Delta Q(\rho)}{\Delta Q(\rho)}\right)\bigl[\Delta_{n,+}^2
\log [\mathcal{H}_{a,u}(\Delta_{n,+}) ]
-\Delta_{n,-}^2
\log [\mathcal{H}_{a,u}(-\Delta_{n,-}) ]\bigr]
\\
&
\quad\quad{}
+\left[
\alpha+\frac{1}{2}+\frac{2+a}{12}\left(1-\frac{\rho\partial_r\Delta Q(\rho)}{\Delta Q(\rho)}\right)\right]
\bigl[\log [\mathcal{H}_{a,u}(\Delta_{n,+}) ]-\log\bigl[\mathcal{H}_{a,u}(-\Delta_{n,-}) ]\bigr].
\end{align*}
By regularizing the integral, we obtain
\begin{align*}
& \frac{1}{12}\int_{-\Delta_{n,-}}^{\Delta_{n,+}}
\frac{\widetilde{g}_{a,u}(x)}{g_{a,u}(x)}\,{\rm d}x
=
-\frac{a}{24}\left(1+\frac{2\rho\partial_r\Delta Q(\rho)}{\Delta Q(\rho)}\right)\bigl(\Delta_{n,+}^2-\Delta_{n,-}^2\bigr)
\\
&\quad{}
+\left(\alpha+\frac{1}{2}\right)
\log [\mathcal{H}_{a,u}(\Delta_{n,+}) ]
-
\left(\alpha+\frac{1}{2}\right)
\log [\mathcal{H}_{a,u}(-\Delta_{n,+}) ]
\\
&\quad{}
+
\frac{2+a}{12}\left(1-\frac{\rho\partial_r\Delta Q(\rho)}{\Delta Q(\rho)}\right)
\log [\mathcal{H}_{a,u}(\Delta_{n,+}) ]\\
& \quad
-
\frac{2+a}{12}\left(1-\frac{\rho\partial_r\Delta Q(\rho)}{\Delta Q(\rho)}\right)
\log [\mathcal{H}_{a,u}(-\Delta_{n,-}) ]
\\
&\quad{}
-
\frac{1}{12}\left(2+\frac{\rho\partial_r\Delta Q(\rho)}{\Delta Q(\rho)}\right)
\Delta_{n,+}^2\log [\mathcal{H}_{a,u}(\Delta_{n,+}) ]\\
&\quad
+
\frac{1}{12}\left(2+\frac{\rho\partial_r\Delta Q(\rho)}{\Delta Q(\rho)}\right)
\Delta_{n,-}^2\log [\mathcal{H}_{a,u}(-\Delta_{n,-}) ]
\\
&\quad{}
+\frac{1}{6}\left(2+\frac{\rho\partial_r\Delta Q(\rho)}{\Delta Q(\rho)}\right)
\int_{-\infty}^{+\infty}
\left[
x\left(\log [\mathcal{H}_{a,u}(x) ]
- u\mathbf{1}_{(-\infty,0)}
\right)
-ax\log|x|
-\frac{a(a-1)x}{2(x^2+1)}
\right] {\rm d}x
\\
&\quad{}
-\frac{u}{12}\left(2+\frac{\rho\partial_r\Delta Q(\rho)}{\Delta Q(\rho)}\right)
\Delta_{n,-}^2
+\frac{1}{6}\left(2+\frac{\rho\partial_r\Delta Q(\rho)}{\Delta Q(\rho)}\right)
\int_{-\Delta_{n,-}}^{\Delta_{n,+}}\!
\left(ax\log|x| + \frac{a(a-1)x}{2(x^2+1)}\right) {\rm d}x
\\
&\quad{}
+\mathcal{O}\left(\frac{(\log n)^{1/4}}{n^{1/4}}\right).
\end{align*}
Combining all of the above together with \eqref{def of asymptotics of mathcal H au x}, we obtain the result.
\end{proof}

\subsection{Proof of Theorem~\ref{theorem:calEn expansion}}
\label{subsection:last proof}
We now complete the proof of Theorem~\ref{theorem:calEn expansion}.
Combining the error terms from Lemmas~\ref{lemma:S1},~\ref{lemma:Analysis of S2}, and~\ref{lemma:S3}, we have $C_n^{(1)}+C_{n}^{(2)}+C_{n}^{(3)}=
-a\log2-\frac{a}{2}\log\Delta Q(\rho)$, which is added to the term of order~$\mathcal{O}(1)$.
Thus, by Lemmas~\ref{lemma:S0},~\ref{lemma:S1},~\ref{lemma:Analysis of S2}, and~\ref{lemma:S3}, we get $C_3^{(1)}+C_3^{(2)}+C_3^{(3)}-a\log2-\frac{a}{2}\log\Delta Q(\rho)=C_3(u,a)$.
By Lemmas~\ref{lemma:S0},~\ref{lemma:S1},~\ref{lemma:Analysis of S2}, and~\ref{lemma:S3}, $C_1(u,a),C_2(u,a)$ are similarly computed.
This completes the proof of Theorem~\ref{theorem:calEn expansion}.

\appendix
\section[Consistency between Theorem 1.3 and \protect{[40, Theorem 1.1]}]{Consistency between Theorem~\ref{theorem:asymptotic expansion of counting statistics} and \cite[Theorem~1.1]{C2021-FH}}
\label{section: appendix counting}
In this appendix, we confirm the consistency between Theorem~\ref{theorem:asymptotic expansion of counting statistics} and \cite[Theorem~1.1]{C2021-FH} mentioned in Remark~\ref{remark:counting statistics}.
We recall that $C_1$, $C_2$, $C_3$ for $Q(z)=|z|^{2b}$ ($b>0$) $r_1=\rho$, and $r_{k}\equiv0$ for $k=2,3,\dots,m+1$ in \cite[Theorem 1.1]{C2021-FH} are given by
\begin{gather}
\label{def of C1 count MF}
C_1=b\rho^{2b}u,
\\
\label{def of C2 count MF}
C_2=\sqrt{2}b\rho^b\int_0^{+\infty}\bigl(
\mathcal{F}(t,{\rm e}^u)+\mathcal{F}(t,{\rm e}^{-u})
\bigr)\,{\rm d}t,
\\
\label{def of C3 count MF}
C_3=-\left(\frac{1}{2}+\alpha\right)u+4b\int_0^{+\infty}t\bigl(\mathcal{F}(t,{\rm e}^u)-\mathcal{F}(t,{\rm e}^{-u})\bigr)\,{\rm d}t
+b\int_{-\infty}^{+\infty}\mathcal{G}(t,{\rm e}^u)\frac{5t^2-1}{3}\,{\rm d}t,
\end{gather}
where $\mathcal{F}(t,s)$ for $t\in\R$ and $s\in\C\backslash(-\infty,0]$ is given by \eqref{def of mathcalFts} and
\begin{equation*}
\mathcal{G}(t,s):= \frac{1-s}{1+\frac{s-1}{2}\mathrm{erfc}(t)}\frac{{\rm e}^{-t^2}}{\sqrt{\pi}} =\frac{{\rm d}}{{\rm d}t}\mathcal{F}(t,s).
\end{equation*}
It is straightforward to see that $C_{1}(u)$ and $C_2(u)$ in Theorem~\ref{theorem:asymptotic expansion of counting statistics} are consistent with~\eqref{def of C1 count MF} and~\eqref{def of C2 count MF}, respectively.
To see that $C_3$ in \cite[Theorem~1.1]{C2021-FH} is consistent with~\eqref{def of C3 count MF}, by integration by parts, the last term in~\eqref{def of C3 count MF} can be rewritten as
\begin{align*}
&\int_{-\infty}^{+\infty}\mathcal{G}(t,{\rm e}^u)\frac{5t^2-1}{3}\,{\rm d}t\\
&\qquad {}=
 b \lim_{M\to+\infty}\int_0^{M}\mathcal{G}(t,{\rm e}^u)\frac{5t^2-1}{3}\,{\rm d}t
-
b\lim_{M\to+\infty}\int_0^{M}\mathcal{G}(t,{\rm e}^{-u})\frac{5t^2-1}{3}\,{\rm d}t
\\
&\qquad{} =\frac{b}{3}u-\frac{10b}{3}\int_0^{+\infty}t\bigl( \mathcal{F}(t,{\rm e}^{u})-\mathcal{F}(t,{\rm e}^{-u})\bigr)\,{\rm d}t.
\end{align*}
Therefore, \eqref{def of C3 count MF} can be deduced to
\[
C_3=-\left(\frac{1}{2}+\alpha\right)u+\frac{b}{3}u+\frac{2b}{3}\int_0^{+\infty}t\bigl(\mathcal{F}(t,{\rm e}^{u})-\mathcal{F}(t,{\rm e}^{-u})\bigr)\,{\rm d}t.
\]
If we informally apply Theorem~\ref{theorem:asymptotic expansion of counting statistics} to the case $Q(z)=|z|^{2b}$ ($b>0$), by $2+\frac{\rho\partial_r\Delta Q(\rho)}{\Delta Q(\rho)}=2b$, we obtain $C_3(u)=C_3$.
This shows the consistency between Theorem~\ref{theorem:asymptotic expansion of counting statistics} and \cite[Theorem~1.1]{C2021-FH}.

\section[Consistency between Theorem 1.6 and \protect{[28, Theorem 1.1]}]{Consistency between Theorem~\ref{theorem:calEn expansion} and \cite[Theorem~1.1]{BC2022}}

\label{section:appendix parabolic cylinder function}
We confirm the consistency between Theorem~\ref{theorem:calEn expansion} and
\cite[Theorem~1.1]{BC2022}.
We begin with collecting some facts on the parabolic cylinder function and the associated Hermite polynomials from \cite{BC2022, NIST,Temme,Wunsche}.
The $\nu$-th associated Hermite polynomials $\bigl\{\mathrm{He}_k^{(\nu)};k=0,1,\dots\bigr\}$ are defined recursively by
\begin{gather*}
 \mathrm{He}_{k+1}^{(\nu)}(x)=x\mathrm{He}_{k}^{(\nu)}(x)-(k+\nu)\mathrm{He}_{k-1}^{(\nu)}(x), \qquad k\geq 1, \\
 \mathrm{He}_{0}^{(\nu)}(x)=1,\qquad \mathrm{He}_{1}^{(\nu)}(x)=x,
\end{gather*}
and satisfy the orthogonality relations
\begin{equation*}
\int_{-\infty}^{+\infty}\mathrm{He}_{k}^{(\nu)}(x)\mathrm{He}_{\ell}^{(\nu)}(x)\,\frac{{\rm d}x}{|D_{-\nu}({\rm i}x)|^2}=\sqrt{2\pi}(k+\nu)!\delta_{k,\ell},
\end{equation*}
see \cite[equation~(12.7.2)]{NIST}, where the parabolic cylinder function $D_{-\nu}(x)$ is given by \eqref{def of parabolic cylinder function}.
It is known that the parabolic cylinder function is related to a family of Hermite polynomials $\{H_n\}_{n\in\mathbb{N}}$, i.e.,
\begin{equation*}
 D_{n}(z)={\rm e}^{-\frac{1}{4}z^2}2^{-n/2}H_n\left(\frac{z}{\sqrt{2}}\right)
 ={\rm e}^{-\frac{1}{4}z^2}\mathrm{He}_n(z), \qquad n=0,1,2,\dots,
\end{equation*}
where $H_n(x)=2^{n/2}\mathrm{He}_n\bigl(\sqrt{2}x\bigr)$ is defined by
\[
H_n(x)=(-1)^n{\rm e}^{x^2}\frac{{\rm d}^n}{{\rm d}x^n}{\rm e}^{-x^2},
\]
see \cite[equation~(18.5.5)]{NIST}.
For $\nu\in\R$ with $\nu\notin\mathbb{Z}_{<0}$ and $n\in\mathbb{Z}_{\geq 0}$,
\begin{equation*}
D_{-\nu-n}(z):=\frac{\Gamma(-\nu-n+1)}{\Gamma(-\nu+1)}
\bigl[
(-{\rm i})^n\mathrm{He}_{n}^{(\nu-1)}({\rm i}z)D_{-\nu}(z)-(-{\rm i})^{n-1}\mathrm{He}_{n-1}^{(\nu)}({\rm i}z)D_{-\nu+1}(z)
\bigr].
\end{equation*}
Particularly,
\begin{equation*}
D_0(z):=U\bigl(-\tfrac{1}{2},z\bigr)={\rm e}^{-\frac{z^2}{4}},\qquad
D_{-1}(z):=U\bigl(\tfrac{1}{2},z\bigr)={\rm e}^{\frac{z^2}{4}}\sqrt{\frac{\pi}{2}}\mathrm{erfc}\left(\frac{z}{\sqrt{2}}\right).
\end{equation*}
If $a\in\mathbb{Z}_{>0}$ and $\nu=1$, then by $\lim_{\nu\to1}\frac{\Gamma(-\nu-a+1)}{\Gamma(-\nu+1)}=\frac{(-1)^a}{a!}$, one can write
\[
{\rm e}^{-\frac{z^2}{4}}D_{-1-a}(z)
=\frac{(-1)^a}{a!}
\left[
\frac{1}{{\rm i}^{a}}\mathrm{He}_{a}^{(0)}({\rm i}z)\sqrt{\frac{\pi}{2}}\mathrm{erfc}\left(\frac{z}{\sqrt{2}}\right)
-\frac{1}{{\rm i}^{a-1}}\mathrm{He}_{a-1}^{(1)}({\rm i}z){\rm e}^{-\frac{z^2}{2}}
\right].
\]
Next, we recall some functionals from \cite[equations~(1.8), (1.9), (1.10), (1.11), (1.13), and (1.14)]{BC2022} (here we mainly focus on $a\geq 3$ for the simplicity),
\begin{gather*}
 p_{0,a}(x):=\frac{1}{{\rm i}^a} \mathrm{He}_{a}({\rm i}x), \qquad q_{0,a}(x):=\frac{1}{{\rm i}^{a-1}}\mathrm{He}_{a-1}^{(1)}({\rm i}x), \\
 p_{1,a}(x):=-\frac{a}{2}p_{0,a+1}(x)-ab \left(p_{0,a+1}(x)-(3a-1)p_{0,a-1}(x)+\frac{5}{3}(a-1)(a-2)p_{0,a-3}(x)\right),
 \\
 q_{1,a}(x):=-\frac{a}{2}q_{0,a+1}(x)-ab \left(q_{0,a+1}(x)-(3a-1)q_{0,a-1}(x)+\frac{5}{3}(a-1)(a-2)q_{0,a-3}(x)\right),
 \\
 \mathcal{G}_{0}(y;u,a):=p_{0,a}\bigl(-\sqrt{2}y\bigr)\left((-1)^a+\frac{{\rm e}^u-(-1)^a}{2}\mathrm{erfc}(y)\right)+q_{0,a}\bigl(-\sqrt{2}y\bigr)({\rm e}^u-(-1)^a)\frac{{\rm e}^{-y^2}}{\sqrt{2\pi}},
 \\
 \mathcal{G}_{1}(y;u,a):=p_{1,a}\bigl(-\sqrt{2}y\bigr)\left((-1)^{a}+\frac{{\rm e}^u-(-1)^a}{2}\mathrm{erfc}(y)\right)
 +q_{1,a}\bigl(-\sqrt{2}y\bigr)({\rm e}^{u}-(-1)^{a})\frac{{\rm e}^{-y^2}}{\sqrt{2\pi}}.
\end{gather*}
The fact that $\mathrm{He}_{a}^{(\nu)}(-{\rm i}y)=(-1)^a\mathrm{He}_{a}^{(\nu)}({\rm i}y)$ for $y\in\R$ and $a\in\mathbb{N}$ gives rise to
\begin{align*}
&
{\rm e}^{u}\frac{\Gamma(a+1)}{\sqrt{2\pi}}{\rm e}^{-\frac{y^2}{4}}D_{-1-a}(y)
+\frac{\Gamma(a+1)}{\sqrt{2\pi}}{\rm e}^{-\frac{y^2}{4}}D_{-1-a}(-y)
\\
&\qquad =
\frac{1}{{\rm i}^{a}}\mathrm{He}_{a}^{(0)}(-{\rm i}y)
\left[
(-1)^a+({\rm e}^{u}-(-1)^a)\frac{1}{2}\mathrm{erfc}\left(\frac{y}{\sqrt{2}}\right)
\right]\\
&\quad \qquad{} +
 ({\rm e}^{u}-(-1)^{a} )\frac{1}{{\rm i}^{a-1}}\mathrm{He}_{a-1}^{(1)}(-{\rm i}y)\frac{{\rm e}^{-\frac{y^2}{2}}}{\sqrt{2\pi}}.
\end{align*}
This shows the consistency between $\mathcal{H}_{a,u}(y)$ given by \eqref{def of cal H au} and \cite[equation~(1.13)]{BC2022}, i.e.,
\[
\mathcal{G}_{0}\left(\frac{y}{\sqrt{2}};u,a\right)
=
 {\rm e}^{u}\frac{\Gamma(a+1)}{\sqrt{2\pi}}{\rm e}^{-\frac{y^2}{4}}D_{-1-a}(y)
+\frac{\Gamma(a+1)}{\sqrt{2\pi}}{\rm e}^{-\frac{y^2}{4}}D_{-1-a}(-y)
=
\mathcal{H}_{a,u}(y).
\]
By a similar manner to the above, since by \cite[equation~(12.8.1)]{NIST},
\begin{gather*}
D_{2-a}(y)=\bigl(y^2+a-1\bigr)D_{-a}(y)+ayD_{-a-1}(y), \\
(a+1)D_{-a-2}(y)=-yD_{-a-1}(y)+D_{-a}(y),
\end{gather*}
we have
\begin{gather*}
\mathcal{G}_{1}\left(\frac{y}{\sqrt{2}};u,a\right)
 =
-\frac{\Gamma(a+1)}{\sqrt{2\pi}}{\rm e}^{-\frac{y^2}{4}}
\bigg[
\left(\frac{a}{2}+ab \right)(a+1) \bigl({\rm e}^{u}D_{-2-a}(y)-D_{-2-a}(-y) \bigr)
\\
\hphantom{\mathcal{G}_{1}\left(\frac{y}{\sqrt{2}};u,a\right)=}{}
-b(3a-1)\bigl({\rm e}^uD_{-a}(y)-D_{-a}(-y)\bigr)
+\frac{5}{3}b\bigl({\rm e}^u D_{2-a}(y)-D_{2-a}(-y)\bigr)
\bigg]
\\
\hphantom{\mathcal{G}_{1}\left(\frac{y}{\sqrt{2}};u,a\right)}{}
=
-\frac{\Gamma(a+1)}{\sqrt{2\pi}}{\rm e}^{-\frac{y^2}{4}}
\biggl[
\left(\frac{a}{2}+ab-b(3a-1)+\frac{5}{3}b(y^2+a-1)\right)\\
\hphantom{\mathcal{G}_{1}\left(\frac{y}{\sqrt{2}};u,a\right)=}{}\times
\bigl({\rm e}^{u}D_{-a}(y)-D_{-a}(-y)\bigr)
\\
\hphantom{\mathcal{G}_{1}\left(\frac{y}{\sqrt{2}};u,a\right)=}{}
+\left(\frac{2}{3}ab-\frac{a}{2}\right)\bigl({\rm e}^{u}D_{-a-1}(y)+D_{-a-1}(-y)\bigr)
\biggr].
\end{gather*}
This gives rise to
\begin{gather}
\frac{\mathcal{G}_{1}\bigl(\tfrac{y}{\sqrt{2}};u,a\bigr)}{\mathcal{G}_{0}\bigl(\tfrac{y}{\sqrt{2}};u,a\bigr)}
=
-\left[
\frac{a}{2}+ab-b(3a-1)+\frac{5}{3}b(y^2+a-1)
\right]\partial_y\bigl(\log [\mathcal{H}_{a,u}(y) ]-u\mathbf{1}_{(-\infty,0)}(y)\bigr)\nonumber\\
\hphantom{\frac{\mathcal{G}_{1}\bigl(\tfrac{y}{\sqrt{2}};u,a\bigr)}{\mathcal{G}_{0}\bigl(\tfrac{y}{\sqrt{2}};u,a\bigr)}
=}{}
-\left(-\frac{a}{2}+\frac{2}{3}ab\right)y,\label{def of wG1 over wG1}
\end{gather}
where we have used \eqref{def of d relationship between g0 g1}.
By substituting \eqref{def of wG1 over wG1} into \cite[$\tfrac{\mathcal{G}_1(y;u,a)}{\mathcal{G}_0(y;u,a)}$ of $C_3$ in Theorem 1.1]{BC2022} and straightforward computations, we find that \cite[$C_3$ in Theorem 1.1]{BC2022} recovers $C_3$ in Theorem~\ref{theorem:calEn expansion}.

\section{Euler--Maclaurin formula}
In this work, we have used Euler--Maclaurin formula \cite{AFLS25,ACC2023c, BKS2023} to establish the precise large $n$-asymptotics of \eqref{def of mathcal En}, which we state below for reference.
\begin{Lemma}[Euler--Maclaurin formula]\label{lemma:Euler--Maclaurin}
Let $f(x)$ be $2m$ times differentiable function on the interval $[p,q]$. Then we have
\begin{align*}
&\sum_{j=p+1}^{q-1}f(j)=\int_{p}^{q}f(x)\,{\rm d}x-\frac{f(p)+f(q)}{2}+\sum_{j=1}^{m-1}\frac{B_{2j}}{(2j)!}\bigl(
f^{(2j-1)}(q)-f^{(2j-1)}(p)
\bigr)
+R_{2m},
\\
&\sum_{j=p}^{q}f(j)=\int_{p}^{q}f(x)\,{\rm d}x+\frac{f(p)+f(q)}{2}+\sum_{j=1}^{m-1}\frac{B_{2j}}{(2j)!}\bigl(
f^{(2j-1)}(q)-f^{(2j-1)}(p)
\bigr)
+R_{2m},
\\
&\sum_{j=p}^{q-1}f(j)
=\int_{p}^{q}f(x)\,{\rm d} x-\frac{f(q)-f(p)}{2}+\sum_{j=1}^{m-1}\frac{B_{2j}}{(2j)!}\bigl(f^{(2j-1)}(q)-f^{(2j-1)}(p)\bigr)+R_{2m},
\end{align*}
where the sequence $\{B_{2k}\}_k$ are even indexed Bernoulli numbers $B_{2}=\frac{1}{6}$, $B_4=-\frac{1}{30},\dots$, and the remainder term $R_{2m}$ satisfies the bound $|R_{2m}|\leq c_{2m}\int_{p}^{q}|f^{(2m)}(x)|\,{\rm d}x$ with $c_{2m}=\frac{2\zeta(2m)}{(2\pi)^{2m}}$.
Here, $\zeta(s)$ is the Riemann zeta function.
\end{Lemma}

\subsection*{Acknowledgements}

The author is grateful to Sung-Soo Byun and Seong-Mi Seo for valuable feedbacks.
The author acknowledges support from the European Research Council (ERC), Grant Agreement No. 101115687.
I~thank the referees for the careful reading and additional comment, which has improved our presentation.

\pdfbookmark[1]{References}{ref}
\LastPageEnding

\end{document}